\documentclass[sigconf]{acmart}

\renewcommand\footnotetextcopyrightpermission[1]{} % removes footnote with conference information in first column
\usepackage{booktabs} % For formal tables
\usepackage{mathrsfs}
\usepackage{kbordermatrix}
\usepackage{xspace}

\newcommand{\etal}{\hbox{\emph{et al.}}\xspace}

\newcommand\clu{{\sc cl}$_u$ }
\newcommand\cln{{\sc cl}$_n$ }

\begin{document}

\title{Doric: Foundations for Statistical Fault Localisation}

%\titlenote{Produces the permission block, and copyright information}

\author{David Landsberg}
\affiliation{\institution{University College London}}
\email{d.landsberg@ucl.ac.uk}

\author{Earl T. Barr}
\affiliation{\institution{University College London}}
\email{e.barr@ucl.ac.uk}

\renewcommand{\shortauthors}{D. Landsberg \etal}

\begin{abstract}

To fix a software bug, you must first find it. As software grows in size and
  complexity, finding bugs is becoming harder. To solve this problem, measures
  have been developed to rank lines of code according to their "suspiciousness"
  wrt being faulty.  Engineers can then inspect the code in descending order of
  suspiciousness until a fault is found.  Despite advances, ideal measures ---
  ones which are at once lightweight, effective, and intuitive --- have not yet
  been found. We present \textsc{Doric}, a new formal foundation for statistical
  fault localisation based on classical probability theory.  To demonstrate
  \textsc{Doric}'s versatility, we derive \textsc{cl}, a lightweight measure of
  the likelihood some code caused an error.  \textsc{cl} returns
  probabilities, when spectrum-based heuristics (\textsc{sbh}) usually return difficult
  to interpret scores.  \textsc{cl} handles fundamental fault scenarios that
  spectrum-based measures cannot and can also meaningfully identify causes with
  certainty. We demonstrate its effectiveness in, what is to our knowledge, the
  largest scale experiment in the fault localisation literature.  For Defects4J
  benchmarks, \textsc{cl} permits a developer to find a fault after inspecting 6
  lines of code 41.1.8\% of the time.  Furthermore, \textsc{cl} is more accurate
  at locating faults than all known 127 \textsc{sbh}.  In particular, on Steimann's
  benchmarks one would expect to find a fault by investigating 5.02 methods, as
  opposed to 9.02 with the best performing \textsc{sbh}.

\end{abstract}




\maketitle

% dont use so many parantheses
% limit updating. 
% 2.  fault name get rid of it!
% 4. brag about size of benchmarks - intro, setup, and keep in threats. 

\section{Introduction}

Software fault localisation is the problem of quickly identifying the parts of the code that caused an error. Accordingly, 
the development of effective and efficient
methods for fault localisation has the potential to greatly reduce costs,
wasted programmer time, and the possibility of catastrophe~\cite{newscientist}.
In this paper, we focus on methods of lightweight statistical software fault localisation. In general, statistical methods use a given fault localisation measure to assign lines of code a real number, called that line of code's "suspiciousness" degree, as a function some statistics about the program and test suite. In spectrum-based fault localisation, the engineer then inspects the code in descending order of suspiciousness until a fault is found. 
The driving force behind research in spectrum-based fault localisation is the search for an "ideal" measure. 

What is the ideal measure? We assume it should satisfy three properties.  First,
the measure should be effective at finding faults. A measure is effective if an
engineer would find a fault more quickly using the measure than not.  Following
Parnin and Orson, and in the absence of user trials to validate it, we assume
that experiment can estimate a measure's effectiveness by determining how often
a fault is within the top "handful" of most suspicious lines of code under the
measure~\cite{PO11}.  Second, the measure should be lightweight. A measure is
lightweight if an algorithm can compute it
fast enough that an impatient developer does not lose interest.
The current gold standard in speed is spectrum-based,
whose values usually take seconds to compute
and scale to large programs~\cite{7390282}.
Third, an ideal measure should compute meaningful values that
describe more than simply which lines of code are more/less
"suspicious" than others. A canonical meaningful value is
the likelihood, under probability theory, that the given code was faulty.

Debugging is an instance of the scientific method:  developers observe,
hypothesise about causes, experiment by running code, then \emph{crucially}
update their hypotheses.  \textsc{Doric} allows the definition of fault
localisation measures that model this process --- measures that we can update in
light of new data.  In \autoref{section_method}, we present $\textsc{cl}_u$, a
method for updating our \textsc{cl} measure that does just this.

To advance the search for an ideal measure, we propose a ground-up re-foundation of statistical fault localisation based on probability theory.
The contributions of this paper are as follows:

\begin{itemize}

  \item We propose \textsc{Doric}: a new formal foundation for statistical fault
    localisation.  Using \textsc{Doric}, we derive a causal likelihood
    measure and integrate it into a novel localisation method.

\item We provide a new set of fundamental fault scenarios which, we argue, any statistical fault localisation method should analyze correctly. We show our new method does so, but that no {\sc sbh} can. 

\item We demonstrate the effectiveness of \textsc{cl} in, what is to our
  knowledge, the largest-scale fault localisation experiment to date:
    \textsc{cl} is more accurate than all 127 known {\sc sbh}s,
    and when a developer
    investigates only 6 non-faulty lines, no {\sc sbh} outperforms it on Defects4J where the developer would find a fault 41.18\% of the time.

\end{itemize}

All of the tooling and artefacts needed to reproduce our results are available
at \url{utopia.com}.

\section{Preliminaries}\label{Preliminaries}

To reconstruct statistical fault localisation ({\sc sfl}) from the ground up, we must precisely define our terms. {\sc sfl}  conventionally assumes a number of artifacts are available. This includes a program (to perform fault localisation on), a test suite (to test the program on), and some units under test located inside the program (as candidates for a fault)~\cite{Steimann:2013:TVV:2483760.2483767}. 
From these, we define 
coverage matrices, the formal object at the heart of many statistical fault localisation techniques~\cite{7390282}, including our own. 

%\subsection{Faulty Programs}

\textbf{Faulty Programs}.
Following Steimann \etal's terminology~\cite{Steimann:2013:TVV:2483760.2483767}, a \emph{faulty program}
is a program that fails to always satisfy a \emph{specification}, which is a
property expressible in some formal language and describes the intended
behavior of some part of the program.  When a
specification fails to be satisfied for a given execution (i.e., an
\emph{error} occurs), we assume there exists some lines of code in the program that cause the error for that execution, identified as
the \emph{fault} (aka \emph{bug}). 

\begin{example}
An example of a faulty {\sc c} program is given in Fig.~\ref{Code 1}
({\tt minmax.c}), taken from Groce \etal~\cite{Groce04errorexplanation}).
We use it as our running example throughout this paper. 
Some executions of {\tt minmax.c} violate the specification {\tt least
<= most} (i.e., there are some executions where there is an error). Accordingly, in these executions, a corresponding assertion (the last line of the program) is violated. Thus, the program fails to always satisfy the
specification.  The fault in this example is labeled {\tt u3}, which
should be an assignment to {\tt least} instead of {\tt
most}. \label{faultyprogs}
\end{example}

\begin{figure}[t!]
\begin{minipage}[t]{0.5\linewidth}
\noindent

\begin{verbatim}
int main() { 

  int in1, in2, in3; 
  int least = in1; 
  int most = in1; 
  
  if (most < in2)     
    most = input2; // u1    
      
  if (most < in3)     
    most = in3; // u2
    
  if (least > in2)     
    most = in2; // u3 (fault)  
    
  if (least > in3)     
    least = in3; // u4   
 
  assert(least <= most) 
  }   
  
\end{verbatim}\vspace{-1cm}~\caption{{\tt minmax.c.}}\label{Code 1} 
\noindent
    
\end{minipage}\hfill
\begin{minipage}[t]{0.5\linewidth}

\centering

\vspace{0.8cm}

\kbordermatrix{\mbox{}& u_1 & u_2 & u_3 & u4 & e\\
t_1 & 0 & 1 & 1 & 0 & 1\\
t_2 & 0 & 0 & 1 & 1 & 1\\
t_3 & 0 & 0 & 1 & 0 & 1\\
t_4 & 1 & 0 & 0 & 1 & 0\\
t_5 & 1 & 1 & 0 & 0 & 0\\
}\caption{Coverage Matrix.}\label{Coverage Matrix}

\vspace{1.1cm}

\begin{tabular}{l l l}
 & Vector & Oracle \\
\hline
$t_1$ & $\langle 1, 0, 2 \rangle$ & fail \\
$t_2$ & $\langle 2, 0, 1 \rangle$ & fail \\
$t_3$ & $\langle 2, 0, 2 \rangle$ & fail \\
$t_4$ & $\langle 1, 2, 0 \rangle$ & pass \\
$t_5$ & $\langle 0, 1, 2 \rangle$ & pass
\end{tabular}\caption{Test Suite.}\label{Test Suite}

\end{minipage} 
\vspace{-0.3cm}
\end{figure}

\medskip

%\subsection{Test Suites}
%if(least > most) 
 %   error = 1; // u5 (error)

\textbf{Test Suites}. 
Each program has a set of test cases called a test suite $T$. Following Steimann \etal~\cite{Steimann:2013:TVV:2483760.2483767}, a \textit{test case} 
is a repeatable execution of some part of a program. We assume each test case is associated with an input vector to the program and some oracle which describes whether the test case fails or passes. A test case \emph{fails} if, by executing the program on the test case's input vector, the resulting execution violates a given specification, and \emph{passes}  otherwise.

\begin{example}
The test case associated with input vector $\langle 0, 1, 2 \rangle$ is an execution in
which {\tt in1} is assigned 0, {\tt in2} is assigned 1, and {\tt
in3} is assigned 2, the {\sc uuts} labelled {\tt u1}, {\tt u2}  are executed, but the {\sc uuts} labelled {\tt u3} and {\tt u4}, are not executed. As the specification {\tt least <= most} is satisfied in that execution (i.e. there is no violation to the assertion statement), an error does not occur.  For
the running example we assume a test suite exists consisting of five test cases $T = \{t_1, \dots  ,t_{5}\}$. Each test case is associated with an input vector and oracle as described in Figure~\ref{Test Suite}. For our example of {\tt minmax.c}, the oracle is an error report which tells the engineer in a commandline message whether the assertion has been violated or not. 

\label{testcase}
\end{example}

\textbf{Units Under Test}. 
A \emph{unit under test} ({\sc uut}) is a concrete artifact in a given program. Intuitively, a {\sc uut} can be thought of as a candidate for being faulty. The collection of {\sc uut}s is chosen by software engineer, according to their requirements.  Many types of {\sc uut}s have been used in the literature, including
methods~\cite{DBLP:conf/issre/SteimannF12},
blocks~\cite{Abreu:2006:ESC:1193217.1194368,
DiGiuseppe:2011:IMF:2001420.2001446}, branches~\cite{5070508}, and
statements~\cite{Jones:2002:VTI:581339.581397, journals/jss/WongQ06,
Liblit:2005:SSB:1064978.1065014}.  A~{\sc uut} is said to be \emph{covered}
by a test case if that test case executes the {\sc uut}.   
Notationally, we define a set of \textit{units} as $U$. For notational
convenience in the definition of coverage matrics, $U$ also contains a special
unit $e$, called the \textit{error}, that a test case covers if it fails. We let
$U^* = U - \{e\}$ and $U_{|U|} = e$.

\begin{example}
In Figure~\ref{Code 1},
the {\sc uut}s are the statements labeled in comments
marked {\tt u1}, $\dots$, {\tt u4}. Accordingly, the set of units is  $U = \{u_1, u_2, u_3, u_4, e\}$. \label{uuts}
\end{example}

\textbf{Coverage Matrices}. A useful way to represent the coverage details of a test suite is in the form of a coverage matrix. It will first help to introduce some notation. For a  matrix $c$, we let $c_{i,k}$ be the value of the $i$th column and $k$th row of $c$.

\begin{definition}~\label{def_coverage_matrices}
A \textit{coverage matrix} is a Boolean matrix $c$ of height $|T|$ and width $|U|$, where for each $u_i \in U$ and $t_k \in T$: 

\begin{equation*}
c_{i,k}={\begin{cases}
\:\:1 \:\:\:\:\: \textrm{if} \  t_k  \ \textrm{covers} \ u_i\\
\:\:0 \:\:\:\:\: \textrm{otherwise} \end{cases}}
\end{equation*}

\end{definition}

We abbreviate $c_{|U|,k}$ with $e_k$. 
Intuitively, for all $u_i \in U^*$, $c_{i,k} = 0$ just in case $t_k$ executed $u_i$, and 0 otherwise. $e_{k} = 1$ just in case $t_k$ fails, and 0 otherwise. 
We use the notational abbreviations of Souza etal~\cite{DBLP:journals/corr/SouzaCK16}.
$\sum_k$ is $\sum_{k = 1}^{|T|}$.
$\sum_{k} c_{i,k}e_k $ is $c_{ef}^i$. Intuitively, this is the number of test cases that \underline{e}xecute $u_i$ and \underline{f}ail.
$ \sum_{k} \overline{c}_{i,k} {e}_k $ is $c_{ef}^i$. Intuitively, this is the number of test cases that do \underline{n}ot execute $u_i$ but \underline{f}ail.
$ \sum_{k} {c}_{i,k} \overline{e}_k $ is $c_{ep}^i$. Intuitively, this is the number of test cases that \underline{e}xecute $u_i$ and \underline{p}ass.
$\sum_k \overline{c}_{i,k} \overline{e}_k$ is $c_{np}^i$. Intuitively, this is
the number of test cases that do \underline{n}ot execute $u_i$ but
\underline{p}ass. When the context is clear, we drop the leading $c$ and the index $i$. 
A coverage matrix for the running example is given in Fig.~\ref{Coverage Matrix}.

%Naish et al.~\cite{Naish:2011:MSS:2000791.2000795} and Wong et al.~\cite{Wong:2007:EFL:1299135.1299726}. 
%\subsection{Spectrum Based Functions}~\label{subsection_sbfl}

\begin{table}[t]
\begin{tabular}{lc}
Name & Expression  \\
\hline
\\

 Ochiai &  $\frac{ef^2}{(ef + nf)(ef + ep)}$ \vspace{0.2cm} \\
 
 D3 & $\frac{ef^3}{ep + nf}$ \vspace{0.2cm} \\
 
 Zoltar & $\frac{ef}{ef + nf + ep + \frac{10000nfep}{ef}}$ \vspace{0.2cm}  \\
 
 GP05 & $\frac{(ef + np)\sqrt{ef}}{(ef + ep)(npnf + \sqrt{ep})(ep + np ) \sqrt{| ep - np |}}$ \vspace{0.2cm} \\ 
 
%  Added-value &$\max(\frac{ef}{ef+ep} - \frac{ef+nf}{|T|}, \frac{ef}{ef+nf} - \frac{ef+ep}{|T|})$ \vspace{0.2cm} \\
  
%  Klosgen & $\sqrt{\frac{ef}{|T|}} \times$ Added-value \vspace{0.2cm} \\
  
    Naish & $ef - \frac{ep}{ep+np+1}$ 
  
\end{tabular}
\vspace{0.3cm}
\caption{Some Suspiciousness Functions}\label{esf} 
\vspace{-0.5cm}

\end{table}

\textbf{Spectrum-Based Heuristics}. One way to measure the suspiciousness of a given unit wrt how faulty it is is to use a spectrum-based heuristic, sometimes called a spectrum-based "suspiciousness" measure~\cite{DBLP:journals/corr/SouzaCK16}.

\begin{definition}~\label{def_sbh}
 A \textit{spectrum-based heuristic} ({\sc sbh}) is a function $s$ with signature $s : U^* \rightarrow \mathbb{R}$.  For each $u \in U^*$ $s(u_i)$ is called $u_i$'s degree of \textit{suspiciousness}, and is defined as a function of $u_i$'s \textit{spectrum}, which is the vector  $\langle c_{ef}^i,c_{nf}^i,c_{ep}^i,c_{np}^i \rangle$.
\end{definition}

 The intuition behind {\sc sbh}'s is that $s(u_i) > s(u_j)$ just in case $u_i$ is more "suspicious" wrt being faulty than $u_j$. 
In \textit{spectrum-based fault localisation} ({\sc sbfl}), {\sc uuts} are inspected by the engineer in descending order of suspiciousness until a fault is found. When two units are equally suspicious some tie-breaking method is assumed. One method is choosing the unit which appears earlier in the code to inspect first. We shall assume this method in this paper. 

We discuss a property of some {\sc sbh}'s.
If a suspiciousness function $s$ is \textit{single fault optimal} then for all $u_i, u_j \in U$, if $c_{ef}^i = c_{ef}^i + c_{nf}^i$ and $c_{ef}^i > c_{ef}^j$, then $s(u_i) > s(u_j)$.
Intuitively, this states if a measure is single-fault optimal then {\sc uut}s executed by all failing traces are more suspicious than ones that aren't~\cite{Naish:2011:MSS:2000791.2000795, Landsberg}. This property is based on the observation that the fault will be executed by all failing test cases in a program with only one fault. An example of a single fault optimal measure is the Naish measure (see Table~\ref{esf}). 

 %Secondly, for all $c > 0$, If $ef = ef + nf$ and $np = c$, then the value returned is greater than any value returned when $np < c$.

\begin{example}
To illustrate how an {\sc sbh} can be used in {\sc sbfl}, we perform {\sc sbfl} with Wong-II measure $s(u_i)$ = $c_{ef}^i - c_{ep}^i $ = $\sum_k c_{i,k}e_k - \sum_k c_{i,k}\overline{e}_k$~\cite{Wong:2007:EFL:1299135.1299726} on the running example. $s(u_1)$ = -2, $s(u_2)$ = 0, $s(u_3)$ = 3, and $s(u_4)$ = 2. Thus the most suspicious {\sc uut} ($u_3$) is successfully identified with the fault. Accordingly, in a practical instance of {\sc sbfl} the fault will be investigated first by the engineer. 
\end{example}

\section{\textsc{Doric}: New Foundations}\label{section_classical_foundations}

We present \textsc{Doric}\footnote{Given our goal of
providing a simple foundation to statistical fault localisation, we name our
framework after this simple type of Greek column}, our formal framework based on
probability theory. We proceed
in four steps. First, we define a set of models to represent the universe of
possibilities. Each model represents a possible way the error could have been
caused (Section~\ref{section_models}). Second, we define a syntax to express
hypotheses, such as "the $i$th uut was a cause of the error" and a semantics
that maps a hypothesis to the set of models where it is true
(Section~\ref{section_syntax}).  Third, we outline a general theory of
probability (Section~\ref{section_probability}). Then, we develop a classical
interpretation of probability (Section~\ref{section_classical_interpretation}).
Using this interpretation, we define a measure usable for fault localisation
(Section~\ref{section_ecp}). Finally, we present our fault localisation
methods~\ref{section_method}.

\subsection{The Models of Doric}~\label{section_models}

\begin{figure}[t!]
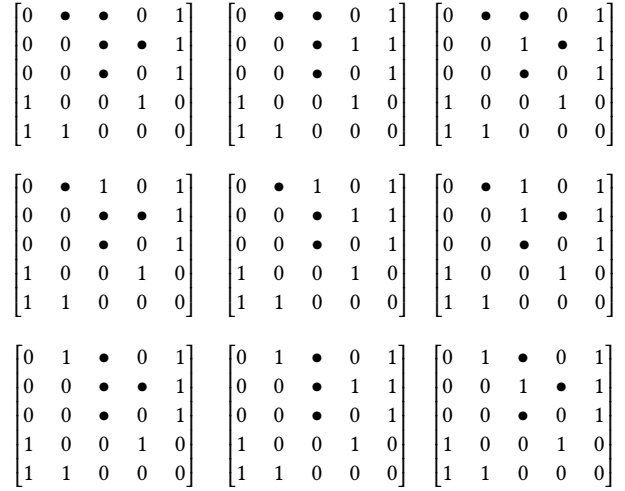


% 1st ROW of MATRICES
% {\setlength{\fboxsep}{0pt}\colorbox{white}{0}} 
% {\setlength{\fboxsep}{0pt}{\setlength{\fboxsep}{0pt}\colorbox{lightgray}{1}} } 

\[\begin{bmatrix}
0 & \bullet   & \bullet  & {\setlength{\fboxsep}{0pt}\colorbox{white}{0}} & 1 \\ 
0 & 0 & \bullet  & \bullet & 1 \\
0 & 0 & \bullet & 0 & 1 \\
1 & 0 & 0 & 1 & 0 \\
1 & 1 & 0 & 0 & 0 
\end{bmatrix} \
\hspace{0.25cm}
\begin{bmatrix}
{\setlength{\fboxsep}{0pt}\colorbox{white}{0}} & \bullet  & \bullet  & {\setlength{\fboxsep}{0pt}\colorbox{white}{0}} & 1 \\ 
0 & 0 & \bullet  & 1 & 1 \\
0 & 0 & \bullet  & 0 & 1 \\
1 & 0 & 0 & 1 & 0 \\
1 & 1 & 0 & 0 & 0 
\end{bmatrix}
\hspace{0.25cm}
\begin{bmatrix}
{\setlength{\fboxsep}{0pt}\colorbox{white}{0}} & \bullet  & \bullet  & {\setlength{\fboxsep}{0pt}\colorbox{white}{0}} & 1 \\ 
0 & 0 & 1 & \bullet  & 1 \\
0 & 0 & \bullet  & 0 & 1 \\
1 & 0 & 0 & 1 & 0 \\
1 & 1 & 0 & 0 & 0 
\end{bmatrix}
\hspace{0.25cm}
\]

% 2nd ROW of MATRICES

\[\begin{bmatrix}
{\setlength{\fboxsep}{0pt}\colorbox{white}{0}} & \bullet  & {\setlength{\fboxsep}{0pt}\colorbox{white}{1}} & {\setlength{\fboxsep}{0pt}\colorbox{white}{0}} & 1 \\ 
0 & 0 & \bullet   & \bullet  & 1 \\
0 & 0 & \bullet  & 0 & 1 \\
1 & 0 & 0 & 1 & 0 \\
1 & 1 & 0 & 0 & 0 
\end{bmatrix} \
\hspace{0.25cm}
\begin{bmatrix}
{\setlength{\fboxsep}{0pt}\colorbox{white}{0}} & \bullet  & {\setlength{\fboxsep}{0pt}\colorbox{white}{1}} & {\setlength{\fboxsep}{0pt}\colorbox{white}{0}} & 1 \\ 
0 & 0 & \bullet   & 1 & 1 \\
0 & 0 & \bullet  & 0 & 1 \\
1 & 0 & 0 & 1 & 0 \\
1 & 1 & 0 & 0 & 0 
\end{bmatrix}
\hspace{0.25cm}
\begin{bmatrix}
{\setlength{\fboxsep}{0pt}\colorbox{white}{0}} & \bullet  & 1 & {\setlength{\fboxsep}{0pt}\colorbox{white}{0}} & 1 \\ 
0 & 0 & 1 & \bullet  & 1 \\
0 & 0 & \bullet & 0 & 1 \\
1 & 0 & 0 & 1 & 0 \\
1 & 1 & 0 & 0 & 0 
\end{bmatrix}
\hspace{0.25cm}
\]

% 3rd ROW of MATRICES

\[\begin{bmatrix}
{\setlength{\fboxsep}{0pt}\colorbox{white}{0}} & {\setlength{\fboxsep}{0pt}\colorbox{white}{1}} & \bullet  & {\setlength{\fboxsep}{0pt}\colorbox{white}{0}} & 1 \\ 
0 & 0 & \bullet    & \bullet   & 1 \\
0 & 0 & \bullet   & 0 & 1 \\
1 & 0 & 0 & 1 & 0 \\
1 & 1 & 0 & 0 & 0 
\end{bmatrix} \
\hspace{0.25cm}
\begin{bmatrix}
{\setlength{\fboxsep}{0pt}\colorbox{white}{0}} & {\setlength{\fboxsep}{0pt}\colorbox{white}{1}} & \bullet   & {\setlength{\fboxsep}{0pt}\colorbox{white}{0}} & 1 \\ 
0 & 0 & \bullet   & 1 & 1 \\
0 & 0 & \bullet   & 0 & 1 \\
1 & 0 & 0 & 1 & 0 \\
1 & 1 & 0 & 0 & 0 
\end{bmatrix}
\hspace{0.25cm}
\begin{bmatrix}
{\setlength{\fboxsep}{0pt}\colorbox{white}{0}} & {\setlength{\fboxsep}{0pt}\colorbox{white}{1}} & \bullet   &
{\setlength{\fboxsep}{0pt}\colorbox{white}{0}} & 
 1 \\ 
0 & 0 & 1 & \bullet    & 1 \\
0 & 0 & \bullet   & 0 & 1 \\
1 & 0 & 0 & 1 & 0 \\
1 & 1 & 0 & 0 & 0 
\end{bmatrix}
\hspace{0.25cm}
\]

~\caption{Causal Models.}~\label{Causal Models}
\vspace{-0.5cm}
\end{figure}

In our framework, classical probabilities are defined in terms of the proportion of models in which a given formula is true. To achieve this, we first define a set of models for our system. We first describe some notation used in the forthcoming definition of models here. Let $j \in \mathbb{N}$ and $m^j$ be a matrix, then $m^j_{i,k}$ is the value of the cell located at the $i$th column and $k$th row of matrix $m^j$, where this value is in $\{1, 0,\bullet \}$. As with coverage matrices, the rows represent test cases and the columns represent units. Informally, for each cell $m^j_{i,k}$, 0 denotes $u_i$ was neither executed by $t_k$ nor a cause of the error $e$, 1 denotes $u_i$ was executed but was not a cause of $e$, and $\bullet$ denotes $u_i$ was executed by $t_k$ and was a cause of $e$. 

\begin{definition}\label{def_causal_models}
Let T be a test suite and U a set of units. The set of \textit{models} for a coverage matrix $c$ is a set of matrices $\mathcal{M}$ = $\{m^1,\dots,$ $m^{|\mathcal{M}|}\}$ of height $|T|$ and width $|U|$ satisfying:

\begin{equation*}
m^j_{i,k} \in {\begin{cases}
\:\: \{1,\bullet\} \:\:\:\:\: \textrm{if} \ c_{i,k}e_k = 1 \: \& \: i \neq |U| \\
\:\: \{c_{i,k}\}  \:\:\:\:\ \textrm{otherwise} \\ \end{cases}}
\end{equation*}

where for each $t_k \in T$ there is some $u_i \in U$ such that $m^j_{i,k}$ = $\bullet$.

\end{definition}

Informally, each model (also called a \textit{causal model}) describes a possible scenario in which errors were caused. The scenario is epistemically possible --- logically possible and (we have assumed) consistent with what the engineer knows. 
Underlying our definition are three assumptions about the nature of causation. First, causation is factive: if a unit causes an error in a given test case, then the {\sc uut} has to be executed and the error both have to factually obtain for a causal relation to hold between them. Second, errors are caused: if an error occurs in a given test case, then the execution of some {\sc uut} caused it. Third, causation is irreflexive: no error causes itself.

\begin{example}
The set of causal models $\mathcal{M}$ of the running example is given in Fig.~\ref{Causal Models}. Following Def.~\ref{def_causal_models}, there are 9 models $\{m^1, \dots, m^9 \}$. $\mathcal{M}$ represents all the different combinations of ways {\sc uuts} can be said to be a cause of the error in each test case. In Fig.~\ref{Causal Models}, we associate $m^1$, $m^2$, $m^3$ with the top three models, $m^4$, $m^5$, $m^6$ with the middle three models, and $m^7$, $m^8$, $m^9$ with the bottom three models. 

%For the purposes of visual clarity "2" is uniformaly replaced in each cell with a grey highlighted "1".
%There are $|M| = |M_1| \times |M_2| \times |M_3|  \times |M_4| \times |M_5|= (2^2-1) \times (2^2-1) \times (2^1-1) \times 1 \times 1 = $ 9 models.  For the purposes of visual clarity "2" is uniformaly replaced in each cell with a grey highlighted "1". % For the purposes of space we have not included models in which there is no causal hypothesis for some test case (in fact for the running example 

\end{example}

\subsection{The Syntax and Semantics of Doric}~\label{section_syntax}

What sort of hypotheses does the engineer want to estimate the likelihood of? In
this section, we present a language fundamental to the fault localisation task.
This language includes hypotheses about which line of code was faulty, which
caused the error in which test case, etc. We develop such a language as follows.
First, we define a set of basic \textit{partial causal hypotheses} $H = \{h_1,
\dots, h_{|U|}\}$, where $h_i$ has the reading "the $i$th {\sc uut} was
\textit{a} cause of the error". Second, we define set of \textit{basic
propositions}  $U = \{u_i,\dots,u_{|U|}\}$, where $u_i$ here takes a
propositional reading "the $i$th {\sc uut} was executed". 

\begin{definition}\label{def_language}
$L$ is called the \textit{language}, defined inductively over a given set of basic propositions $U$ and causal hypotheses $H$ as follows:
\begin{enumerate}
%\item $U = \{u_i,\dots,u_{|U|}\}$, is the set of units.
%\item $H = \{h_1,\dots,h_{|U|}\}$, is the set of atomic hypotheses
\item if $\phi \in U \cup H$, then $\phi \in L$
\item if $\phi, \psi \in L$, then $\phi \wedge \psi, \neg \psi, \Diamond_k \phi \in L$, for each $t_k \in T$
\end{enumerate}
\end{definition}

%$\phi \rightarrow \psi$ is $\neg\phi \vee \psi$. 

We use the following abbreviations and readings .
$\phi \vee \psi $ abbreviates $\neg(\neg\phi \wedge \neg\psi)$, read "$\phi$ or $\psi$".
$\phi \wedge \psi$ is read "$\phi$ and $\psi$". 
$\phi \vee \psi$ is read "$\phi$ or $\psi$".
$\neg \phi$ is read "it is not the case that $\phi$". 
$\Diamond_k \phi$ is read "$\phi$ in the kth test case".
$e$ is read "the error occurred".  In addition, we define $L^*$, called the \textit{basic language}, which is a subset of $L$ defined as follows: If $u_i \in U$ then $u_i \in L*$, if $\phi, \psi \in L$, then $\phi \wedge \psi, \neg \phi \in L$. 

An important feature of $L$ is that we abbreviate two additional types of hypotheses, as follows: First, 
$H_i$ is $h_i \bigwedge_{h_j \in H - \{h_i\}} \neg h_j$, where $H_i$ is read "the $i$th {\sc uut} was \textit{the} cause of the error", is called a \textit{total causal hypothesis} for the error, and intuitively abbreviates the property that the $i$th unit was a cause of the error, and nothing else was. 
Second, we let
$f_i$ is abbreviated $\bigvee_{t_k \in T} \Diamond_k h_i$ , where $f_i$ is read "the $i$th {\sc uut} is a fault", is called a \textit{fault hypothesis}, and intuitively abbreviates the property that the $i$th unit was a cause of the error in some test case.

We now treat the semantics of Doric. 
To determine which propositions in the language $L$ are true in which models $M$, we provide valuation functions mapping propositions to models, as follows:

\begin{definition}~\label{def_semantics}
Let $M$ be a set of models and let $L$ be the language. Then the set of \textit{valuations} is a set $V = \{v_1,\dots, v_{|V|}  \}$, where for each $t_k \in T$ there is some $v_k \in V$ with signature $v_k : L \rightarrow 2^M$, defined inductively as follows:
\begin{enumerate}
\item $v_k(u_i) = \{m^j \in M| m_{i,k}^j \in \{1, \bullet \}\}$, for $u_i \in U$
\item $v_k(h_i) = \{m^j \in M| m_{i,k}^j = \bullet \}$, for $h_i \in H$  
\item $v_k(\phi \wedge \psi) = v_k(\phi) \cap v_k(\psi)$, for $\phi, \psi \in L$
\item $v_k(\neg \phi)$ = $M - v_k(\phi)$, for $\phi \in L$
\item $v_k(\Diamond_n \phi)$ = $v_n(\phi)$, for $\phi \in L$
\end{enumerate}
\end{definition}

$v_k(\phi)$ is read "the models where $\phi$ in the $k$th test case". 
We give an example to illustrate.

\begin{example}
We continue with the running example. Each of the following can be visually verified by checking the causal models in Figure~\ref{Causal Models}. $v_1(u_1)$ = $\emptyset$. Intuitively, this is because $t_1$ executes $u_1$ in no models. $v_1(u_2) = M$. Intuitively, this is because $t_1$ executes $u_1$ in all models. $v_1(h_2) = \{m^1, \dots, m^6\} = M - \{m^7, m^8, m^9\}$. Accordingly, $u_1$ was a cause of the error in 6 out of 9 models. In contrast, $v_1(H_2) = v_1(h_2 \wedge \neg h_1 \wedge \neg h_3 \wedge \neg h_4) = \{m^4, m^5, m^6\}$. Accordingly, $u_2$ was \textit{the} cause of the error in 3 out of 9 models. Finally, $v_1(\diamond_3 h_3) = v_3(h_3) = M$. Accordingly, in the third test case the 3rd {\sc uut} was a cause of the error. 
\end{example}

\subsection{The Probability Theory of Doric}~\label{section_probability}

We want to determine the probability of a given hypothesis. 
We do this by presenting our theory of probability. The theory is based around the following assumptions:
We assume the engineer does not always know which hypotheses are true of each test case. Accordingly, we want our probabilities about hypotheses to take an \textit{epistemic} interpretation, in which the probabilities describe how much a given hypothesis should be believed.

\begin{definition}\label{def_probability}
Let $M, L, V$ be non-empty sets of models, a language and set of valuations respectively. Then, a \textit{probability theory} is a tuple $\langle \textbf{P}, P, w \rangle$, where 

%${\textbf{P}} = \{P_1, \dots, P_{|T|}\}$, where 

\begin{enumerate}
\item $w: 2^M \rightarrow \mathbb{R}$
\item $\textbf{P}$ =  $\{P_1, \dots, P_{|T|}\}$ is a set of \textit{probability functions}, where for each  $v_k \in V$ there is some $P_k \in \textbf{P}$ such that 

\begin{equation}
P_{k}(\phi) = \frac{w(v_k(\phi))}{w(M)}, {\ \textrm{for} \ \textrm{each} \ } \phi \in L
\end{equation}

\item $P$ is the \textit{expected likelihood} function defined as follows:

\begin{equation}
P(\phi) = \sum_{k} \frac{P_{k}(\phi)}{|T|}, {\ \textrm{for} \ \textrm{each} \ } \phi \in L
\end{equation}
\end{enumerate}

\end{definition}

The \textit{weight} function is $w$ and describes the relative likelihood of a set of models.
$P_k(\phi)$ is the probability that $\phi$ holds in the $k$th test case, and is defined as the proportion of models in which $\phi$ holds. 
$P(\phi)$ is the expected likelihood that $\phi$, and is defined as the average probability that $\phi$ holds in a test case. 
We use the following readings:
$w(X)$ is "the relative likelihood of the models in X".
$P_k(\phi)$ is "the probability that $\phi$ in the $k$th test case".
$P(\phi)$ is "the expected likelihood that $\phi$". 
We use the standard abbreviation of $P(\phi|\psi)$ for $P(\phi \wedge \psi)/P(\psi)$, which reads "the probability that $\phi$ when $\psi$". 

We now discuss immutable assumptions on $w$.  In order to ensure $P_k$ satisfies standard measure theoretic properties, we assume  $w(\emptyset)$ = 0, $w(M)$ > 0, and $w(X \cup Y) = w(X) + w(Y)$ when $X \cap Y = \emptyset$. We allow any extension to the definition of $w$ satisfying the above properties. When $w$ is so defined, we say it provides an interpretation of the probability functions. For instance, one option is to formally define the relative likelihood of models in terms of the number of faults in them. 

Finally, we establish the intuitive result that the likelihood of a given formula in the basic language is simply the proportion of test cases in which it is true. 
Let $f$ be a function which intuitively measures the frequency in which a proposition is true in a test suite. Defined: $f(\phi) = \sum f_k(\phi)$, where $f_k(\phi) = 1$ if $v_k(\phi) = M$ and 0 otherwise. 
We then have the following result:

\begin{proposition}\label{prop_basic_lang}
For all $\phi \in L^*$, $P(\phi)$ = $\frac{f(\phi)}{|T|}$
\end{proposition}

\begin{proof}
See Appendix.
\end{proof}

Using this result, we can identify many {\sc sbh}'s  with an intuitive probabilistic expression stated within Doric. For example, $P(u_i \wedge e) = c_{ef}^i/|T|$, $P(u_i \wedge \neg e) = c_{ep}^i/|T|$, $P(\neg u_i \wedge \neg e) = c_{np}^i/|T|$, and $P(\neg u_i \wedge  e) = c_{nf}^i/|T|$. Using these four identities alone one can express the 40 {\sc sbh}'s  of Lucia etal.~\cite{LLJTB14}, and the 20 causal and confirmation {\sc sbh}'s of Landsberg at al.~\cite{Landsberg, landsberg2016methods}.

\subsection{Classical Interpretation}~\label{section_classical_interpretation}

What conditions hold on the relative likelihood function? The question here is which causal models are more likely than others. To illustrate our framework, we will impose conditions on the relative likelihood function to give us a classical interpretation of probability. Informally, probability has a classical interpretation if it satisfies the condition that if there are a total of $n$ mutually exclusive possibilities, the probability of one of them being true is 1/$n$~\cite{jaynes03}. The rationale for this is \textit{the principle of indifference} (aka \textit{the principle of insufficient reason}), which states that if there are a total of $n$ mutually exclusive possibilities, and there is not sufficient reason to believe one over the other, then their relative likelihoods are equal. Formally, for all $x, y \in M$, $ w(\{x\}) = w(\{y\})$. This condition is also known to describe a uniform distribution over the set of models. In the remainder of this paper, we assume this condition.  The assumption is sufficient for the following result:

\begin{equation}~\label{eq_classical_prob}
P_{k}(\phi) = \frac{|v_k(\phi)|}{|M|}, {\ \textrm{for} \ \textrm{each} \ } \phi \in L
\end{equation}

\begin{proposition}
Equation~\ref{eq_classical_prob} follows given indifference
\end{proposition}

\begin{proof}
See Appendix.
\end{proof}

Intuitively, the probability of a proposition is the ratio of models in which it is true. 
Equation~\ref{eq_classical_prob} is tantamount to assuming that, \textit{ab initio}, the engineer knows next to nothing about what caused the error in each test case (each causal model is equally likely).
In practice, we think is probably wrong for the purposes of software fault localisation (causal models with a small number of faults are probably more likely). The main reason for its assumption is that it keeps our forthcoming fault localisation methods simple and tractable. 

We now illustrate how we can use the classical interpretation to give us a definition of $P(f_i)$.  Following our readings, $P(f_i)$ measures the likelihood the ith unit was a fault. This describes is the proportion of models where $i$ is a cause of the error in some test case (Intuitively, $P(f_i)$ is the the proportion of models where there is a "$\bullet$" somewhere in the $i$th column of a model). We call this measure a measure of \textit{fault-likelihood}. Accordingly, the assumption of indifference gives us the following result. Let $i \in [0, |U|]$ and $m, t \in [0, |T|]$ be free variables, and let
$\rho^k = {\sum_{u_i \in U^*}} c_{j,k} $, then:

\begin{equation}~\label{f1}
P(f_i) = P(\bigvee\limits_{k = 1}^{|T|} \Diamond_k h_i) 
\end{equation}

\begin{equation}~\label{f2}
P(\bigvee\limits_{k = m}^{|T|} \Diamond_k h_i) = (1 - P_k(h_i)P(\bigvee\limits_{j = k+1}^{|T|} \Diamond_j h_i))  + P_k(h_i)
\end{equation}

\begin{equation}~\label{f3}
P_t(h_i) = \begin{cases}

\frac{2^{\rho^t-1}}{2^{\rho^t}-1} \:\:\:\:\: \textrm{if} \  c_{i,t}e_{t} = 1 \\

0 \:\:\:\:\:\:\:\:\:\:\:\:\:\:\: \textrm{otherwise}

\end{cases} 
\end{equation}

\begin{proposition}
Equations~\ref{f1},~\ref{f2}, \& ~\ref{f3} follow given indifference.
\end{proposition}

\begin{proof}
See Appendix.
\end{proof}

On the running example, the equations can be used to find $P(f_1) =  0$, $P(f_2) = P(f_4) = 2/3$, $P(f_3) = 1$.

\subsection{Measure for Fault Localisation}~\label{section_ecp}

To develop an efficient fault localisation method based on our framework, we need to do two things. First, we need to identify a probabilistic expression which tells us which unit should be investigated first when looking for faults. Second, we need to identify an efficient way to compute this. In this section, we address these issues in turn.

Which unit should be investigated first when looking for faults?
To answer this, it is tempting to answer --- the unit which is the most likely fault ($\max_{u_i\in U^*}$ $P(f_i)$).  However, under the classical interpretation, this will be ineffective for fault localisation as it ignores passing test cases. 
Moreover,
we do not think that this is necessarily the unit which should be investigated first. Rather, we think that the unit which is estimated to have the highest \textit{propensity} to cause the error should be investigated first. To see the difference, we observe that something might have a high fault-likelihood, but simultaneously have a low propensity to cause errors. Think of a rarely executed bug --- we think these will be of less interest to an engineer. 
Accordingly, the measure we should use will be an expression describing this propensity. 

We make two assumptions in our development. First, we make the assumption that it is possible to find \textit{the} cause, as opposed to \textit{a} cause, where finding \textit{the} cause is preferable.  Second, following Popper~\cite{popper2005logic}, we assume the propensity of some $x$ to $y$ is described by probabilistic expressions of the form $P(y|x)$. Accordingly, the propensity of a unit to cause an error can be analogously described in our framework with $P(H_i|u_i)$. Following our readings of this section, this is read "the likelihood a given {\sc uut} was the cause of an error when it was executed". We call this measure a measure of \textit{causal likelihood} (or {\sc cl} for short), and for some cases is able to identity a some faults with certainty, in the sense that if $P(H_i|u_i) = 1$ then $P(f_i) = 1$. Our answer to our question is thus $\max_{u_i\in U^*}$ $P(H_i|u_i)$.

% class imbalance. 

We now address the question of how to compute {\sc cl} efficiently. One option is to generate all the matrices in the set of causal models, and (following our assumption of indifference) find the probability directly by counting models. However, this is intractable in general. A more tractable alternative is to find an expression for $P(H_i|u_i)$ which is stated purely as a function of a given coverage matrix $c$ and is also tractable. We present this in Eq.~\ref{eq 1} and show it follows from the definitions of this section. As follows:
Let $c$ be a given coverage matrix. Let $\rho_k$ abbreviate $\sum_{u_j \in U^*} c_{j,k}$. Informally, $\rho_k$  is the number of units executed by $t_k$. Then for each $u_i \in U^*$:

\begin{equation}~\label{eq 1}
P(H_i|u_i) = \frac{\sum_{k}^{} {{c_{i,k}e_k}^{-2^{\rho_k}-1}}}{\sum_{k}^{} c_{i,k}}
\end{equation}

\begin{proposition}
Equation~\ref{eq 1} follows using the definitions.
\end{proposition}~\label{prop_eq1}

\begin{proof}

To aid in the proof, it will be useful to establish the following equations.
Let $t \in {[1,|T|]}$ and $i \in {[1,|U|]}$, then:

\begin{equation}\label{proof_eq1}
P(H_i|u_i) = \frac{P(H_i \wedge u_i)}{P(u_i)}
\end{equation}

\begin{equation}\label{proof_eq2}
P(H_i \wedge u_i) = P(H_i)
\end{equation}

\begin{equation}\label{proof_eq3}
P(H_i) = \sum_{k} \frac{P_k(H_i)}{|T|}
\end{equation}

\begin{equation}\label{proof_eq4}
 P_t(H_i) = \begin{cases}

{1}/{(2^{\rho_t}-1)} \:\:\:\:\:\:\:\:\: \textrm{if} \: c_{i,t}, e_{t} = 1  \\

0 \:\:\:\:\:\:\:\:\:\:\:\:\:\:\:\:\:\:\:\:\:\:\:\:\:\:\ \textrm{otherwise}

\end{cases} 
\end{equation}

\begin{equation}\label{proof_eq5}
P(u_i) = \sum_{k} \frac{c_{i,k}}{|T|}
\end{equation}
 
We now sketch the proof of these equations. \ref{proof_eq1} follows by the definition of conditional probability. \ref{proof_eq3} follows by def.~\ref{def_probability}. It remains to give the proofs for equations~\ref{proof_eq2},~\ref{proof_eq4} and~\ref{proof_eq5}. As these are longer they are consigned to the appendix in proposition~\ref{a_causal}. 

Finally, it is easily observed that Eq.~\ref{eq 1} holds using equations \ref{proof_eq1}-\ref{proof_eq5}. As follows: $P(H_i|u_i) = P(H_i|u_i)/P(u_i)$ (by Eq.\ref{proof_eq1}). Thus, $P(H_i|u_i) = P(H_i)/P(u_i)$ (by Eq.\ref{proof_eq2}). So, $P(H_i|u_i) = \sum_k P_k(H_i)/|T|/P(u_i)$ (by Eq.\ref{proof_eq3}). Thus, $P(H_i|u_i) = \sum_k (P_k(H_i)/|T|)/ (\sum_k c_{i,k}/|T|)$ (by Eq.\ref{proof_eq5}). So, $P(H_i|u_i) = \sum_k P_k(H_i)/ \sum_k c_{i,k}$ (by cancellation). It remains to show that $P_k(H_i) = c_{i,k}e_k / 2^{\rho_k -1}$. Assume  $c_{i,k}, e_k = 1$. Then by the first condition of Eq.~\ref{proof_eq4} $P_k(H_i) = 1 / 2^{\rho_k -1}$. This is equal to $c_{i,k}e_k / 2^{\rho_k -1}$ by our assumption. Assume it is not the case that $c_{i,k}, e_k = 1$. Then by the second condition of Eq.~\ref{proof_eq4} $P_k(H_i) = 0$. Now, as either $c_{i,k}$ or $e_k$ is 0 (by def.~\ref{Coverage Matrix}), $c_{i,k}e_k  = 0$, and thus $c_{i,k}e_k / 2^{\rho_k -1} = 0$. Thus, $P_k(H_i) = c_{i,k}e_k / 2^{\rho_k -1}$. 
\end{proof}

\begin{example}\label{more complicated example}
To illustrate {\sc cl}, we find $P(H_i|u_i)$ for each of $u_1,\dots ,u_4$ for the the running example of {\tt minmax.c}. 
We begin with $P(H_1|u_1)$.
We begin by evaluating the numerator of Eq.~\ref{eq 1}, which is equal to $(0 x 0) / (2^{2}- 1) + (0 \times 0) / (2^{2}- 1) + (0 \times 0) / (2^{1}- 1) +  (1 \times 0)/(2^{2}- 1) + (1 \times 0)/(2^{2}- 1) = 0$. The denominator is equal to 2. Thus $P(H_2|u_2)$ = $0/2 = 0$. 
We now do $P(H_2|u_2)$. 
The numerator is equal to $(1 \times 1)/(2^{2}-1)$ + $(1\times 0)/(2^{2}-1)$ + $(1 \times 0)/(2^{1}-1)$ + $(0 \times 0)/(2^{2}-1)$ + $(0 \times 1)/(2^{2}-1)$ = $(1 \times 1)/(2^{2}-1)$ = $1/3$. 
The denominator is equal to 2. Thus $P(H_2|u_2)$ = $(1/3)/2 = 1/6$. $P(H_3|u_3)$ = (1/3 + 1/3 + 1) / 3 = $0.\overline{5}$, and $P(H_4|u_4)$ = (1/3)/2 = $0.1\overline{6}$. Accordingly, the fault $u_3$ is estimated to have the highest likelihood of causing an error when executed.  
\end{example}

We now discuss time-complexity. It is observed that the time complexity of computing the value of both fault and causal likelihood ($P(f_i)$ and  $P(H_i|u_i)$ respectively) is a (small) constant function of the size of the given coverage matrix. This makes it comparable to spectrum-based heuristics in terms of efficiency, which also has this property. To answer our stated question explicitly, to compute $\max_{u_i \in U^*} P(H_i|u_i)$ efficiently we can simply compute Equation~\ref{eq 1} for each $u_i \in U$ using the given coverage matrix $c$, and return the unit with the highest likelihood. 

Finally, we illustrate how our measures of fault and causal likelihood provide a small armory of meaningful measures useful to the engineer. Returning the running example, the engineer might assume the principle of insufficient reason and say of the third unit "the probability it is a fault is 1" and "will likely cause the error when executed" (given $P(f_3) = 1$ and $P(H_3|u_3) > \frac{1}{2}$ respectively). 
We think these quantities are more meaningful than what is reported by some of most effective established {\sc sbh}'s customized to {\sc sbfl}. For instance "the GP05 measure reports the third unit to have a suspiciousness degree of 0.5576" does not have the same meaning, or suggest actionability to the engineer in the same way.

%TODO - 1. unclear* how to get meaningful prob district. skewed. for example, it is unclear. Also, P(e|u) doesn't perform well. 

\subsection{Semi-Automated Methods}\label{section_method} 

We now address the question of how to use Eq~\ref{eq 1} in a fault localisation method. We present two such methods. We then compare our methods to {\sc sbfl}. 

Our first method is similar to the {\sc sbfl} procedure discussed in Section~\ref{Preliminaries}. Here, each $u_i \in U$ is associated with a causal likelihood, as determined by using Eq.~\ref{eq 1}. The engineer then inspects uuts in the program in descending order of causal likelihood (also called suspiciousness) until a fault is found. When two units are equally suspicious, the unit higher up in the code is inspected first. We call this procedure \cln (which abbreviated "\underline{c}ausal \underline{l}ikelihood with \underline{n}o updating").

We now present our second method.
We begin with some motivation. To start the fault localisation process, we assume the engineer will want to investigate $\max_{u_i \in U^*}(P(H_i|u_i))$. However, in the course of further investigation about the program, the engineer will discover new facts, symbolized $\phi$ (for some $\phi \in L$). To find faults, the engineer will then want to find the causal likelihood of different units \textit{given} those facts. Accordingly, the unit the engineer should investigate next should be given by the following formula:

\begin{equation}\label{eq_update}
\max_{u_i \in U^*}(P(H_i|u_i \wedge \phi))
\end{equation}

The above motivates our second method, which is described as follows. First, $\phi$ is set to a tautology, and the value of Eq.~\ref{eq_update} is computed. 
Suppose the unit returned by Eq.~\ref{eq_update} is $u_j$. Then, if $u_j$ is inspectd by the engineer, and if found to be faulty,  the search terminates. If not, we assume $\neg h_j$ and compute the value of Eq.~\ref{eq_update} letting $\phi = \neg h_j$. Suppose the unit returned is $u_k$. The process is similar to before --- if $u_j$ is faulty, the search terminates. If not, we assume $\neg h_k$ and compute the value of Eq.~\ref{eq_update} letting $\phi = \neg h_j \wedge \neg h_k$. The search continues in this way until a fault is found. We call this procedure \clu ("\underline{c}ausal \underline{l}ikelihood with \underline{u}pdating"). Characteristic of \clu is that clues discovered throughout the investigation can be used to give us new probabilities.

We now discuss implementation details of the second method.
In a practical implementation, we can use Equation~\ref{eq 1} in the following way. Let $F \subseteq H$ be the set of indices to causal hypotheses known to be false. Then we can compute Eq.~\ref{eq_update} by redefining $\rho_k$ as $\sum_{u_j \in U} c_{j,k} - \sum_{h_n \in F} c_{n,k}$, and use the expression on the $rhs$ of Eq.~\ref{eq 1} to find the value of the argument of Eq.~\ref{eq_update}. Proof of this is included in an extended version of this paper. Secondly, in a practical implementation, we also allowed ourselves to limit the size of $F$ (called an update bound), which represents the number of updates an engineer is willing to make. If an update bound is set and reached, then \clu procedure continues without any further updates.

%\begin{example}
%TODO As a basic illustration, we apply our method to the example of subsection~\ref{updatability}. 
%\end{example}

We now discuss valuable formal properties pertaining to the second method.
We observe that the process of fault localisation is much like a game of hide and seek --- insofar as when an engineer inspects one location for a fault and it is not there, then our estimations for the likelihood it is elsewhere should increase. We think it is desirable for a fault localisation method to satisfy a similar property. Accordingly, in this section we show that conditioning on causal likelihood (as per the method of \clu) satisfies a similar property, as follows:

\begin{proposition}~\label{hsp}
Let $c$ be a coverage matrix where $c_{i,k} = c_{j,k} = 1$ for some $t_k \in T$. Then $P(H_i|u_i) < P(H_i|u_i \wedge \neg h_j)$.
\end{proposition}

\begin{proof}
See Appendix.
\end{proof}

In the remainder of this section, we compare our new methods with {\sc sbfl} in light of two new fundamental fault scenarios which, we argue, any statistical fault localisation method should analyze correctly.

\[\kbordermatrix{\mbox{}& u_1 & u_2 & u_3 & e\\
t_1 & 1 & 1 & 0 & 1 \\
t_2 & 0 & 0 & 1 & 1 \\
}\]

Consider the coverage matrix above. We argue that for any {\sc sbh} $s$ to be adequate for fault localisation it should satisfy the property that $s(u_3) > s(u_2)$. Our reasoning follows from the assumption that, when an error occurs, there is some executed {\sc uut} that was a cause of it. Accordingly,  we can be certain that $u_3$ is a fault --- nothing else could have caused the error in the second test case. However, we are \textit{not} certain that $u_2$ is a fault (as far as we know $u_1$ could have caused the error instead in the first test case).
However, it is impossible for any {\sc sbh} to satisfy $s(u_3) > s(u_2)$, because spectrum for $u_3$ and $u_2$ is the same (i.e. $\langle c_{ef}^2,c_{nf}^2,c_{ep}^2,c_{np}^2 \rangle$ = $\langle c_{ef}^3,c_{nf}^3,c_{ep}^3,c_{np}^3 \rangle$ = $\langle 1,1,0,0 \rangle$ --- see the definition of a spectrum in def.~\ref{def_sbh}). Thus, their suspiciousness is the same. Subsequently, {\sc sbh}s cannot handle this fundamental case.
In constrast, the measure of causal likelihood gets the answer right, as $P(H_3|u_3) = 1$ and $P(H_2|u_2) = 1/3$.

\[\kbordermatrix{\mbox{}& u_1 & u_2 & u_3 & u_4 & e \\
t_1 & 0 & 0 & 1 & 1 & 1 \\
t_2 & 1 & 1 & 0 & 0 & 1 \\
}\]

Now consider the coverage matrix above.
Suppose we begin the fault localisation process, without any prior knowledge about which units are faulty or non-faulty. Now, according to spectrum-based functions, each unit is equally suspicious (as the spectrum for each unit is the same = $\langle 1,1,0,0\rangle$). Suppose, following the {\sc sbfl} method, we choose to investigate $u_1$, and discover it not to be faulty. Accordingly, on the assumption that in every failing test case some executed unit is a cause of the error, we can now be certain that $u_2$ is a fault given this new information. However, we cannot be so certain that either $u_3$ or $u_4$ is a fault (as it might be the case it isn't but the other is). Thus upon learning $u_1$ isn't a fault, degrees of suspiciousness should be updated to make $u_2$ more suspiciousness than $u_3$ and $u_4$. However, {\sc sbfl} is inadequate for fault localisation because it has no facility for updates of this sort.
In contrast, as a consequence of Proposition~\ref{hsp} the \clu method gets it right. At step one, for all $u_i \in U$, $P(H_i|u_i) = 1/3$. Suppose the engineer learns $u_1$ is not faulty. Thus, the engineer evaluates all $u_i \in U - \{u_1\}$ next. $P(H_2|u_2 \wedge \neg h_1) = 1$, and remain $P(H_3|u_3 \wedge \neg h_1) = P(H_3|u_3 \wedge \neg h_1) = 1/3$. 

\section{Empirical Evaluation}~\label{section_experiments}

In this section, we compare the performance of our new methods with all known 127 {\sc sbh}s  on large faulty programs. The goal of the experiment is to establish whether \cln and \clu are effective at fault localisation than {\sc sbh}s . 

%TODO 2-fault assessment. 

\subsection{Setup}

\begin{table}[t!]
\centering
\begin{tabular}{lrrrrrrr}
\textbf{Program} &  \textbf{V} & \textbf{LOC} & \textbf{U}  & \textbf{F} &  \textbf{P} & \textbf{B} \\
  %   & \textbf{Versions}   &    \textbf{KLoc}    &   \textbf{Kloc} &  \textbf{avg.} & \textbf{Tests} & \textbf{avg.} & \textbf{avg.} & \textbf{avg.}   \\
\hline
Chart               & 24/26      & 50k  & 680         & 3.75    & 190.21 & 1.92 \\
Closure              & 77/133     & 83k  & 3,432    &  2.88         & 3,367.40 & 2.52 \\
Math                & 91/106     & 19k  & 346      & 1.74         & 168.80 & 3.10 \\
Time                & 20/27      & 53k  & 1,204    &  2.20  &  2,542.85 & 4.10       \\
Lang                & 50/65      & 6k   & 96       & 2.16          & 98.26 & 2.62  \\
Mockito             & 27/38      & -    & 574      & 3.78     & 744.96 & 2.59
 \\
%\hline
%\textbf{AVG}      & -& - & - & - & -& - &- & 
\end{tabular}
\vspace{0.2cm}
\caption{Statistics about Defects4j Benchmarks}~\label{D4j}
\vspace{-0.7cm}
\end{table}

\begin{table}[t]
\centering
\begin{tabular}{lrrrrr}
\textbf{Program} & \textbf{V}  &    \textbf{M}   & \textbf{U}   &  \textbf{F} & \textbf{P}  \\
%&  \textbf{Ms} & \textbf{avg.} & \textbf{Ts } & \textbf{}  & \textbf{avg.} & \textbf{avg.}  \\
\hline
AC\_Codec\_1.3 &     543      & 265 &   188      &    5.35          &    16.04   \\
AC\_Lang\_3.0 &      599      & 5373 &  1,666    &  4.22    &     44.84        \\
Daikon\_4.6.4 &      352      & 14387 &   157  &  1.66      &   30.06      \\
Draw2d\_3.4.2 &      570      & 3231  & 89     &   6.71  &  60.73           \\
Eventbus\_1.4 &      577      & 859  &   91      &   8.19    &   75.70            \\
Htmlparser\_1.6  &   599      & 3231 &  600     &   41.70   &  379.17         \\
Jaxen\_1.1.5 &       600      & 1689 &  695     &   70.29   &   581.25             \\
Jester\_1.37b  &     411      & 378  &  64       &    5.09   &     22.94                  \\
Jexel\_1.0.0b13   &  537      & 242 &  335       &  23.15    &    261.39                \\
Jparsec\_2.0  &      598      & 1011 &  510     &   13.14     &    293.59              \\
%\hline
%\textbf{AVG}    &  -  & - & - & -        &      &     
\vspace{-0.2cm} 
\end{tabular}\caption{Statistics about Steimann's Benchmarks.}\label{Steiman}
\vspace{-0.7cm}
\end{table}

We first present the benchmarks used in our experiment, then describe the methods compared in the experiment, the methods we used to evaluate the performance of the different methods. Finally,  we present some research questions for our experiment to answer.

We first describe the benchmarks used in our experiments. We use two sets of Java benchmarks, Defects4j and Steimann's. Each set of benchmarks contains different programs, each program is associated with different initial faulty versions, and each faulty version is associated with a test suite (which includes some failing test cases) and a method for identifying the faulty {\sc uuts}. Statistics for the two sets of benchmarks are presented in Tables~\ref{D4j} and~\ref{Steiman} respectively. The first column gives the name of the program, the second the number of initial faulty versions (V), the third the number of tested units (lines of code ({\sc loc}) for Defects4j, methods (M) for Steimann). The fourth gives the (rounded) average number of units represented in a coverage matrix for the initial faulty versions (U). The number of units represented in the matrix is always smaller than the number of tested units, as columns in the coverage matrix were removed in our experiments if the corresponding unit was not executed by a failing test case (such units are assumed non-faulty). The fifth and sixth columns give the average number of failing (F) and passing (P) test cases represented in a coverage matrix for each initial faulty version. in Table~\ref{D4j}, The last column gives the average number of faulty units represented in each coverage matrix for each initial faulty version (B). In the case of Steimann's each initial version always had one injected fault, so we have not included that column for its corresponding table.  Finally, we could not find a reliable source for the number of lines of code for Mockito. We now discuss particular details about the two sets of benchmarks.

The first set of benchmarks are taken from the Defects4J repository. This is a database consisting of Java program versions with real bugs fixed by developers, and are described in detail by Ren{\'e} \etal~\cite{d4j}. 
As the authors confirm, not all of the versions were usable with software fault localisation methods. A version was unusable if it had a list of faults which did not correspond to an executed line of code (and thus the techniques considered in this paper would not be able to find them). This was because the faults in these versions include omissions of code.  Thus, the proportion of usable versions are reported in the Vs column of Table~\ref{D4j} (for instance 24/26 versions were usable for Chart). 
To generate coverage matrices, we used pre-existent code from the Defects4J repository.\footnote{\url{https://github.com/rjust/defects4j}}
%13 cases

The second set are the Steimann benchmarks. This is a database of large Java programs with  injected faults introduced via mutation testing, and are described in detail by Steimann \etal~\cite{Steimann:2013:TVV:2483760.2483767}.  The benchmarks are described in Table~\ref{D4j}. To generate multiple fault versions, the authors also created one thousand 2,4,8,16, and 32 fault versions, which were created by combining fault injections from the original 1-fault versions. This meant there were a total of 50K+ program versions associated with our second set of benchmarks. 
We used pre-existent code to generate coverage matrices for us, provided to use by the compilers of the benchmarks~\cite{Steimann:2013:TVV:2483760.2483767}.

We now discuss the methods we compare in our experiments. 
We wish to compare {\sc sbfl} methods with the methods developed in this paper. 
We describe these in turn.
We include in our comparison 127 different {\sc sbh}s . These measures are described in Landsberg~\cite{landsberg2016methods}, and is an attempt at an exhaustive list of {\sc sbh}s available in the literature\footnote{Following established conventions on avoiding divisions by zero with the {\sc sbh}s, we added 0.5 to each of the elements of a spectrum~\cite{Naish:2011:MSS:2000791.2000795, Landsberg}}.
As a baseline for our comparison, we also compared the constant measure (which returns a constant value). 
We also compare \clu (with an update bound of 20). Finally, we also compare \cln when used as a substitute {\sc sbfl} measure for the {\sc sbfl} procedure discussed in Section~\ref{Preliminaries}. This allowed us to compare the benefits of updating.  

We now discuss our methods of evaluation. 
The methods we compare all follow the same format insofar as one inspects more suspicious units first, with units which feature higher up in the code inspected first in the case of ties. We wish to evaluate a method in terms of how quickly a fault would be found using this approach. Accordingly, 
First, for each coverage matrix, a method's \textit{accuracy} is defined as the number of non-faulty units investigated using the technique until a fault is found. A method's accuracy for a given set of coverage matrices associated with that benchmark is defined as the average accuracy for that set. For a given set of coverage matrices, the most accurate method is the one with the lowest accuracy score. 

The second method of evaluation is as follows. For a given cut-off of $n$ and a given set of coverage matrices, a method's \textit{n-score} is the percentage of times a fault is found by investigating $n$ non-faulty units using the method. For a set of coverage matrices and given $n$, the method with the best $n$-score is the one with the highest $n$-score.  We provide values of $n$ in the range [0,10].  Following the suggestions of Parnin and Orso~\cite{PO11}, we provide this range as we think 10 units is a realistic upper bound on the number of units an engineer will investigate using a given method. Accuracy and $n$-scores off score are both called a method's \textit{scores} in general. For each of our sets of benchmarks, the \textit{overall} accuracy score was the average accuracy over all coverage matrices in that set of benchmarks. 

We now discuss our research questions. 
Accordingly, 
for each of our benchmarks, our experimental setup is designed to help answer the following research questions:

\begin{enumerate}
\item Which method is the most accurate?

\item Do our new techniques have the best $n$-scores for $n \in [0,10]$?
\end{enumerate}

Finally, all of our results can be reproduced. The Steimann benchmarks can be downloaded from \url{http://www.feu.de/ps/prjs/EzUnit/eval/ISSTA13/}. The Defects4J benchmarks can be downloaded from \url{https://github.com/rjust/defects4j}. All the {\sc sbh} 's compared are available in~\cite{Landsberg, landsberg2016methods}. 

\subsection{Results}

\begin{figure}[t!]
\includegraphics[width=0.5\textwidth]{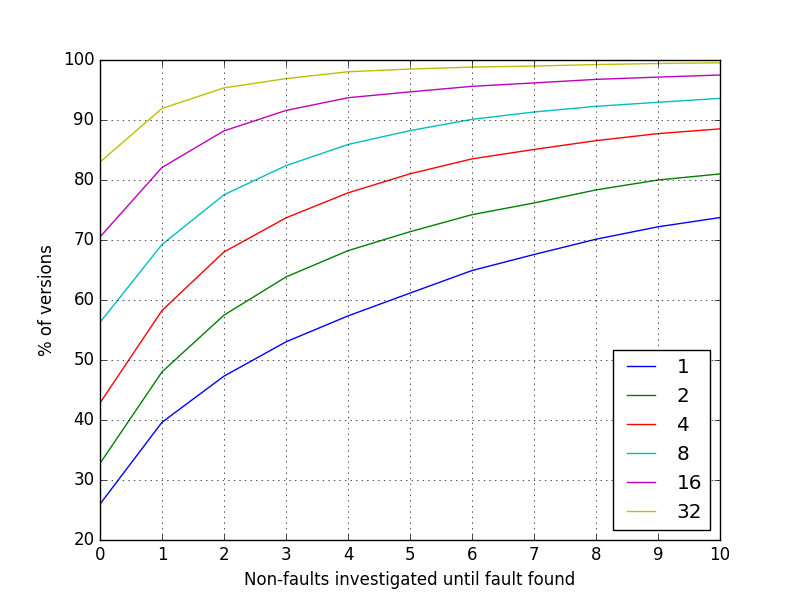}~\caption{Performance of \cln on Steimann's benchmarks with different numbers of faults.}~\label{AB_cactus}
\vspace{-0.8cm}
\end{figure}

\begin{figure}[t!]
\includegraphics[width=0.5\textwidth]{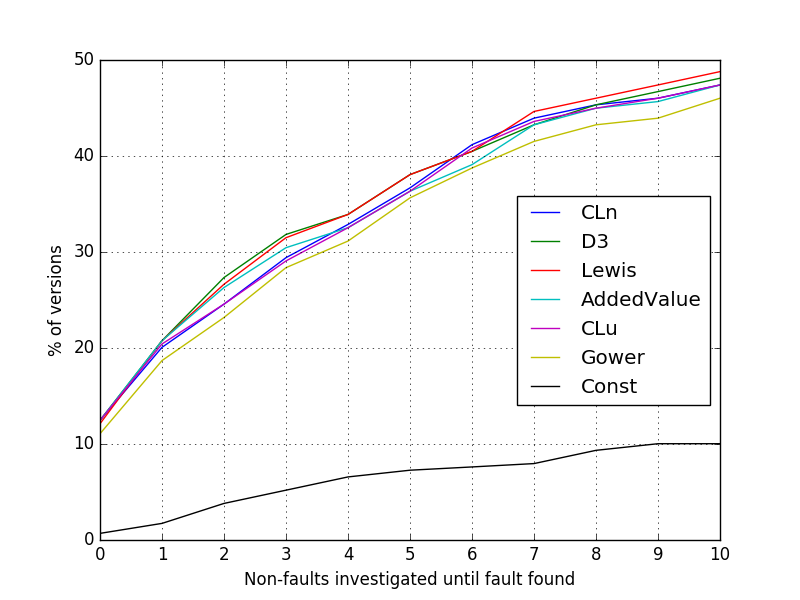}~\caption{ Performance of some methods on Defects4j benchmarks.}~\label{NB_cactus}
\vspace{-0.8cm}
\end{figure}

We first directly answer our two research questions. First,
{Which method is the most accurate?}
For Defects4j \cln has the best accuracy score (213.6).  For  Steimann's  benchmarks, \clu has the best accuracy score (4.9) (\cln came second with 5.02).
Second, {Do our new techniques have the best $n$-scores for any $n \in [0,10]$?}
For Defects4j, \cln had the best 6-score (41.18). For Steimann's benchmarks, for all $n \in [0,10]$, \clu had the best $n$-scores for each of the sets of 4, 8, 16, 32 faults. \cln had the second best scores in these cases.

We first discuss the results for Defects4j, and begin with the overall accuracy scores. To show the range, we report the measures which were ranked 1, 2, 10, 50, 90, 100. These were \cln (213.6), D3 (216.36), Lewis (224.79), AddedValue (232.46), \clu  (241.22) Gower (248.72) respectively. The constant measure was ranked last (643.78).
We now summarise the $n$-scores for each $n \in [0,10]$.
The $n$-scores for a range of methods are presented in Figure~\ref{NB_cactus}. We describe the figure as follows. For each method in the legend, a point was plotted at (x, y) if in y\% of the versions a fault was localized after investigating x non-faulty lines of code. To show the range of performance across methods, the methods in the legend are associated with the aforementioned methods. 
Of all 127 methods compared, \clu did not get the best $n$-score for any of $n \in [0,10]$, and \cln tied with the best 6-score (41.18 - tied with Dennis, Ochiai, and f9830).

We now discuss the results for Steimann's benchmarks. 
We first summarise the overall accuracy scores. To show the range of scores, we report the methods which were ranked 1, 2, 3, 10, 20, 50, 100. These were \clu (4.9), \cln (5.02), Klosgen (9.02), SokalSneath4 (9.51), calf, (9.77), Tarantula (10.54), GP23 (15.68). The constant measure came last (38.69). 
We now summarise accuracy scores on sets of 1, 2, 4, 8, 16, 32 fault programs. On 1-fault programs,  Naish received the best accuracy score (9.81). \clu came 72nd (14.5) and \cln came 74th (14.69). On each of the sets of versions with multiple faults, \clu and \cln came 1st and 2nd respectively, as follows: 
On 2-fault programs, the top 3 were \clu (9.87) \cln and (10.04) and D3  (13.1).
On 4-fault programs, the top 3 were 
\clu (5.14) \cln and (5.3) and GeometricMean (10.34).
On 8-fault programs, the top 3 were 
\clu (2.79) \cln and (2.94) and Klosgen (7.57).
On 16-fault programs, the top 3 were 
\clu (1.16) \cln and (1.24) and AddedValue (6.04).
On 32-fault programs, the top 3 were 
\clu (0.4) \cln and (0.44) and Certainty (4.87).

We now summarise the $n$-scores. For each $n \in [0,10]$, and for the set of 1-fault programs, Naish outperformed \clu or \cln at every value $n$. For each $n \in [0,10]$, and for the set of 2-fault programs, D3  outperformed \clu or \cln at every value of $n$.  For each $n \in [0,10]$, and for each of the sets of 4, 8, 16, 32 fault programs, \clu or \cln outperformed all {\sc sbfl} methods at every value of $n$. \clu and \cln $n$-scores were always similar (+/- 2 percent of on one another), thus to give an indication of how they perform as the number of faults grows in a program, we present Figure~\ref{NB_cactus}. Here, for each of the 1, ..., 32 fault programs, a point was plotted on (x, y) if in y\% of the versions a fault was localized after investigating x non-faulty lines of code.

In the remainder of this section, we discuss our results.
We first discuss differences in value between our evaluative methods (accuracy and $n$-scores). In general, $n$-scores are more important than accuracy from the point of view of an engineer looking for a technique to use in practice. As has been demonstrated in a study by Parnin and Orso, " programmers will stop inspecting statements, and transition to
traditional debugging, if they do not get promising results
within the first few statements they inspect"~\cite{PO11}. $n$-scores are more important from this perspective. 

We now discuss the difference in quality between accuracy and $n$-scores. In general, the accuracy scores for all techniques on Defects4j are poor. For instance, using the most accurate method \cln, one would expect to investigate 213.6 lines of code on average. In contrast, if one limited oneself to investigating fewer than 10 non-faulty lines of code, then one would expect to find a fault almost half the time. For instance, using \cln we can expect to find a fault 41.68\% of the time if one limited oneself to investigating 6 non-faulty lines of code (here, nothing did better than \cln). This suggests that accuracy scores were affected by outliers in which the technique only located a fault after investigating many lines of code (we discuss two such outliers below).  

We now discuss why methods performed differently on Defects4j as opposed to Steimann.
First, in the useable Defects4j versions, the average number of failing test cases way very small, as detailed by Table~\ref{D4j}. In contrast, the number of failing test cases in Steimann benchmarks was much larger. We think this improved the performance of different methods, as more failing tests provide more information about the behavior of the fault. Secondly, in the useable Defects4j versions, 38.4\% had only a single failing test case. Accordingly, \clu ,  \cln and most {\sc sbh}'s performed equivalently on these benchmarks in terms of which units get ranked higher/lower, this meant that the performance of high performing methods tended to converge. 
Thirdly, we emphasize that the {\sc uuts} in Steimann's benchmarks were calls to subclasses, which are often larger than lines of code (the units for Defects4j), and thus the scores look better for Steimann's benchmarks.
Fourthly, for Steimann's benchmarks many of the failing test cases only executed a small part of the overall program. This advantaged our new methods \clu and \clu, which take advantage of short failing executions to increase their causal likelihood. 

%TODO CLASS IMBALANCE problem. will tend to give them lower than they should be. should be representative of the population. 

We now discuss the difference in accuracy between \clu and \cln. As described earlier, the performance of \clu is slightly worse than \cln on Defects4j, where the opposite is observed on Steimann's benchmarks. We investigated reasons why for this, and after investigating the cases where \cln outperforms \clu in terms of accuracy on Defects4j, we discovered there were two major outliers that made \clu 's overall accuracy score lower (in Chart-5 \clu had to investigate 3177 more lines of code, and Math-6 \clu had to investigate Math-6 1585 more lines of code). For Defects4j, \clu outperformed \cln in only two cases, whereas \cln outperformed \clu in 50. These results (tentatively) suggest the conclusion that \cln is better to use in practice when the program and test suite resemble the benchmarks in Defects4j, and \clu when the program and test suite resemble the benchmarks in Steimann's benchmarks.

Finally, we think the scores of our new techniques demonstrate they are a strong contender to {\sc sbfl} heuristics when integrated into a practical fault localisation approach. 

%- mutants can be better at representing the distribution. -- natural faults aren't generally representative of the population - serious bias. 
%We use the second set of benchmarks which contain  injected faults. 

\subsection{Threats to Validity}

Our threats are informed by the recent work of~\cite{d4j}, who perform a similar test on the Defects4j benchmarks, and by Steimann \etal, who perform similar experiments on the Steimann benchmarks~\cite{Steimann:2013:TVV:2483760.2483767}.

The main threat is wrt how well our results generalize to practical instances of fault localisation. Given the variety of programs, faults, test suites, and development styles in "the wild" it has not yet been shown whether studies of this sort generalize well. Wrt the Steimann benchmarks, there is the additional threat that artificial faults are not good proxies for real
faults. Problems of this sort confront many software fault localisation studies~\cite{Steimann:2013:TVV:2483760.2483767}. To anticipate these problems, we have tried to improve the degree to which our results can generalise by ensuring our experiment was large. To our knowledge our experiment is currently the largest in the literature in terms of three different dimensions: number of methods compared (127+), range of faults studied (1-32), and number of program versions used (50k+). 

\vspace{-1mm}
\section{Related Work}~\label{section_related_work}
\vspace{-1mm}
The recent survey of Wong \etal~\cite{Wong:2016:SSF:3012168.3012182} identifies the most prominent fault localisation methods to be spectrum based~\cite{Abreu:2007:ASF:1308173.1308264, Landsberg,  Naish:2011:MSS:2000791.2000795, LLJTB14, Dstar, EricWong:2010:FCC:1672348.1672568, Yoo12, Kim:2015:NHA:2701126.2701207, Santelices:2009:LFU:1555001.1555021}, slice based~\cite{Agrawal:3075077,Zhang:2005:EEU:1085130.1085135, journals/jss/WongQ06,conf/compsac/LeiMDW12, Weiser:1981:PS:800078.802557}, model based~\cite{Mayer:2008:EMM:1642931.1642950, Mayerold, Yilmaz:2007:AMD:1321631.1321659},
and mutation-based ~\cite{Moon:2014:AMM:2624302.2624536, Papadakis:2015:MMF:2858638.2858646}. For reasons of space, we discuss closely related statistical approaches.

We first discuss {\sc sbh}, which is one of the most lightweight methods. In general, {\sc sbh}'s designed to solve the general problem of fault localisation (for programs with any number of faults) are heuristics which estimate how "suspicious" a given unit is. However, what "suspicious"  means is (to our knowledge) never fully defined. In the absence of an approach which tells us this, research is driven
by the development of new measures with improved experimental performance~\cite{Abreu:2007:ASF:1308173.1308264, Landsberg,  Naish:2011:MSS:2000791.2000795, LLJTB14, Dstar, EricWong:2010:FCC:1672348.1672568, Yoo12, Kim:2015:NHA:2701126.2701207, Santelices:2009:LFU:1555001.1555021}. Many of these top-performing measures are presented in Table~\ref{esf}:
 D3 was developed by raising the numerator of a previously used measure to the power of 3~\cite{6258291}. Zoltar was developed by adding $10000\frac{nfep}{ef}$ to the denominator of a previously used measure~\cite{1596502}. GP05 was found using genetic programming~\cite{Yoo12}. Ochiai was originally designed for Japanese fish classification~\cite{Ochiai, Abreu:2007:ASF:1308173.1308264}. In this paper, we have tried to improve upon the theoretical connection between developed measures and the fault localisation problem.

We now discuss theoretical results for {\sc sbfl}.
 Theoretical results
include proving potentially desirable formal properties of measures and finding
equivalence proofs for classes of measures~\cite{Xie, Landsberg, Naish:2011:MSS:2000791.2000795, Naish:Duals, DBLP:conf/sac/DebroyW11}. Yoo \etal have
established theoretical results that show that a "best" performing suspicious
measure for {\sc sbfl} does not exist~\cite{Nopotofgold}, arguing there is "no pot of gold at the end of the program spectrum rainbow" in theory. With this result, there remains the problem of
providing better formal foundations, and deriving improved measures shown to satisfy more formal properties. Our work in this paper is designed to addresses this.

%2. DEEP - Fault localisation Analysis Based on Deep Neural Network 
Two more heavyweight approaches to multiple fault localisation are as follows.  Both of these approaches use fault models in their analysis. Here, a fault model is a set of units with the property that each failing trace executes at least one of them.
The first is the simple classical approach of Steimann \etal~\cite{DBLP:conf/icst/2009}, where the probability of a unit being faulty is the proportion of fault models it is a member of. A second method is Barinel ~\cite{Abreu:2009:SMF:1747491.1747511}, which uses Bayesian analysis and heuristic policies to estimate the health of {\sc uuts}, and uses a tool Staccato~\cite{ Abreu09alow-cost} to generate large sets of fault models.  The main issue confronting these approaches is scalability, given the requirement of generating large sets of fault models. Additionally, in one study the implementation of Barinel was unable to scale to the Steimann benchmarks~\cite{Landsbergpfl}. We have purposely designed our approach to avoid this issue by having a different definition of a model which facilitates tractable fault localisation methods. 

We discuss miscellaneous statistical approaches here. 
One type of approach uses machine learning, including support vector machines and neural network approaches, to perform fault localisation~\cite{4813783, 6058639}. 
However, these approaches do not make the case that a recourse to machine learning methods is necessary and that a principled approach is impossible. 
Other approaches include the crosstab-based method of Wong~\cite{DBLP:journals/tsmc/WongDX12}, the hypothesis testing approach of Liu~\cite{Liu:2006:SDH:1248726.1248773}, and the probabilistic program dependence graph approach of Baah~\cite{Baah2008}. Landsberg \etal provide an axiomatic setup which uses 
{\sc sbh}s in a probabilistic framework, but does not use models~\cite{Landsbergpfl}.  

We now discuss how lightweight statistical methods have been used in the following applications. Firstly, in semi-automated fault localisation in which users inspect
code in descending order of suspiciousness~\cite{PO11}. Secondly, in fully-automated fault
localisation subroutines within algorithms which inductively synthesize (such as
{\sc cegis}~\cite{jhasynt14}) or repair programs (such as {\sc Gen}Prog~\cite{journals/tse/GouesNFW12}).
Thirdly, as a technique combined with other methods~\cite{Kim:2015:NHA:2701126.2701207, xuan:hal-01018935, baudry06a, Ju2014}. Finally, as a potential substitute for heavyweight methods which cannot scale to large programs. Thus, there is a large field of application for the techniques discussed in this paper.  

Finally, for the approach of classical probability, the overview of James~\cite{jaynes03}, the introduction of Tijms~\cite{henk2004understanding}, and the measure-theoretic foundations of Kolmogorov~\cite{kolmogorov1960foundations} all provided bases upon which to develop our approach. 

\vspace{-1mm}
\section{Conclusion}~\label{section_conclusion}
\vspace{-1mm}

In this paper, we have demonstrated there is a principled formal foundation (Doric) available for statistical fault localisation that does not require recourse to spectrum-based heuristics. In general, Doric opens up a world of different meaningful probabilities which can be reported to the engineer to aid in understanding a faulty program. To illustrate the utility of Doric, we developed two lightweight measures of fault and causal likelihood and integrated the latter into our fault localisation method \cln.  In large-scale experimentation, \cln  was demonstrated to be more accurate when compared with all known 127 {\sc sbh}s . In particular, on the Steimann benchmarks \cln was almost twice as accurate as the best performing {\sc sbh} --- you'd expect to find a fault by examining 5.02 methods as opposed to 9.02. \cln also demonstrated to have the highest 6-score on Defects4j. We think the combined effort demonstrates that our measure of causal likelihood is lightweight, effective and maintains a meaningful connection to fault localisation. 

We now discuss directions for future work. 
First, a major step in our work is to experiment with different weight functions ($w$). There are many ways to do this, so it is our hope there will be as much experimentation over different weights as there has been comparing {\sc sbh}s. A natural place to start is to define the relative likelihood of a model as a function of the number of faults in that model in conjunction with some given cumulative distribution function. Following work on fault distributions in software~\cite{grbac2015probability}, we wish to weigh models with a small number of faults to have a higher relative likelihood. The formal development in this paper lays much of the foundations critical for this step.

%Second,  we have found it convenient to use matrices as our models, however, we'd also like to define it in terms of Kripke structures.
%Third, we wish to extend our approach to multiple fault localisation. Finally, we are also in the process of creating a repository where all the coverage matrices and methods are available for experimentation for researchers. 

\newpage

\bibliographystyle{ACM}
\bibliography{mybib}


\begin{thebibliography}{61}

%%% the Web address.

\ifx \showCODEN    \undefined \def \showCODEN     #1{\unskip}     \fi
\ifx \showDOI      \undefined \def \showDOI       #1{#1}\fi
\ifx \showISBNx    \undefined \def \showISBNx     #1{\unskip}     \fi
\ifx \showISBNxiii \undefined \def \showISBNxiii  #1{\unskip}     \fi
\ifx \showISSN     \undefined \def \showISSN      #1{\unskip}     \fi
\ifx \showLCCN     \undefined \def \showLCCN      #1{\unskip}     \fi
\ifx \shownote     \undefined \def \shownote      #1{#1}          \fi
\ifx \showarticletitle \undefined \def \showarticletitle #1{#1}   \fi
\ifx \showURL      \undefined \def \showURL       {\relax}        \fi
% invisible to TeX
\providecommand\bibfield[2]{#2}
\providecommand\bibinfo[2]{#2}
\providecommand\natexlab[1]{#1}
\providecommand\showeprint[2][]{arXiv:#2}

\bibitem[\protect\citeauthoryear{??}{new}{[n. d.]}]%
        {newscientist}
 \bibinfo{year}{[n. d.]}\natexlab{}.
\newblock \bibinfo{title}{{MS Windows NT} Kernel Description}.
\newblock
  \bibinfo{howpublished}{\url{https://www.newscientist.com/gallery/software-faults/}}.
\newblock
\newblock
\shownote{Accessed: 2010-09-30.}


\bibitem[\protect\citeauthoryear{Abreu and van Gemund}{Abreu and van
  Gemund}{2009}]%
        {Abreu09alow-cost}
\bibfield{author}{\bibinfo{person}{Rui Abreu} {and} \bibinfo{person}{Arjan
  J.~C. van Gemund}.} \bibinfo{year}{2009}\natexlab{}.
\newblock \showarticletitle{A Low-Cost Approximate Minimal Hitting Set
  Algorithm and its Application to Model-Based Diagnosis}. In
  \bibinfo{booktitle}{\emph{Abstraction, Reformulation, and Approximation
  (SARA)}}.
\newblock


\bibitem[\protect\citeauthoryear{Abreu, Zoeteweij, and van Gemund}{Abreu
  et~al\mbox{.}}{2006}]%
        {Abreu:2006:ESC:1193217.1194368}
\bibfield{author}{\bibinfo{person}{Rui Abreu}, \bibinfo{person}{Peter
  Zoeteweij}, {and} \bibinfo{person}{Arjan J.~C. van Gemund}.}
  \bibinfo{year}{2006}\natexlab{}.
\newblock \showarticletitle{An Evaluation of Similarity Coefficients for
  Software Fault Localization}. In \bibinfo{booktitle}{\emph{PRDC}}.
  \bibinfo{pages}{39--46}.
\newblock


\bibitem[\protect\citeauthoryear{Abreu, Zoeteweij, and van Gemund}{Abreu
  et~al\mbox{.}}{2007}]%
        {Abreu:2007:ASF:1308173.1308264}
\bibfield{author}{\bibinfo{person}{Rui Abreu}, \bibinfo{person}{Peter
  Zoeteweij}, {and} \bibinfo{person}{Arjan J.~C. van Gemund}.}
  \bibinfo{year}{2007}\natexlab{}.
\newblock \showarticletitle{On the Accuracy of Spectrum-based Fault
  Localization}. In \bibinfo{booktitle}{\emph{TAICPART-MUTATION}}.
  \bibinfo{publisher}{IEEE}, \bibinfo{pages}{89--98}.
\newblock


\bibitem[\protect\citeauthoryear{Abreu, Zoeteweij, and van Gemund}{Abreu
  et~al\mbox{.}}{2009}]%
        {Abreu:2009:SMF:1747491.1747511}
\bibfield{author}{\bibinfo{person}{Rui Abreu}, \bibinfo{person}{Peter
  Zoeteweij}, {and} \bibinfo{person}{Arjan J.~C. van Gemund}.}
  \bibinfo{year}{2009}\natexlab{}.
\newblock \showarticletitle{Spectrum-Based Multiple Fault Localization}. In
  \bibinfo{booktitle}{\emph{ASE}}. \bibinfo{pages}{88--99}.
\newblock


\bibitem[\protect\citeauthoryear{Agrawal, Horgan, London, and Wong}{Agrawal
  et~al\mbox{.}}{1995}]%
        {Agrawal:3075077}
\bibfield{author}{\bibinfo{person}{H. Agrawal}, \bibinfo{person}{J.~R. Horgan},
  \bibinfo{person}{S. London}, {and} \bibinfo{person}{W.~E. Wong}.}
  \bibinfo{year}{1995}\natexlab{}.
\newblock \showarticletitle{{Fault localization using execution slices and
  dataflow tests}}.
\newblock \bibinfo{journal}{\emph{Software Reliability Engineering, 1995.
  Proceedings., Sixth International Symposium on}}, \bibinfo{pages}{143--151}.
\newblock


\bibitem[\protect\citeauthoryear{Ascari, Araki, Pozo, and Vergilio}{Ascari
  et~al\mbox{.}}{2009}]%
        {4813783}
\bibfield{author}{\bibinfo{person}{L.~C. Ascari}, \bibinfo{person}{L.~Y.
  Araki}, \bibinfo{person}{A.~R.~T. Pozo}, {and} \bibinfo{person}{S.~R.
  Vergilio}.} \bibinfo{year}{2009}\natexlab{}.
\newblock \showarticletitle{Exploring machine learning techniques for fault
  localization}. In \bibinfo{booktitle}{\emph{2009 10th Latin American Test
  Workshop}}. \bibinfo{pages}{1--6}.
\newblock


\bibitem[\protect\citeauthoryear{Baah, Podgurski, and Harrold}{Baah
  et~al\mbox{.}}{2008}]%
        {Baah2008}
\bibfield{author}{\bibinfo{person}{George~K. Baah}, \bibinfo{person}{Andy
  Podgurski}, {and} \bibinfo{person}{Mary~Jean Harrold}.}
  \bibinfo{year}{2008}\natexlab{}.
\newblock \showarticletitle{The Probabilistic Program Dependence Graph and Its
  Application to Fault Diagnosis} \emph{(\bibinfo{series}{ISSTA '08})}.
  \bibinfo{pages}{189--200}.
\newblock
\showISBNx{978-1-60558-050-0}


\bibitem[\protect\citeauthoryear{Baudry, Fleurey, and {Le~Traon}}{Baudry
  et~al\mbox{.}}{2006}]%
        {baudry06a}
\bibfield{author}{\bibinfo{person}{Benoit Baudry}, \bibinfo{person}{Franck
  Fleurey}, {and} \bibinfo{person}{Yves {Le~Traon}}.}
  \bibinfo{year}{2006}\natexlab{}.
\newblock \showarticletitle{Improving Test Suites for Efficient Fault
  Localization}. In \bibinfo{booktitle}{\emph{ICSE}}. \bibinfo{publisher}{ACM},
  \bibinfo{pages}{82--91}.
\newblock


\bibitem[\protect\citeauthoryear{de~Souza, Chaim, and Kon}{de~Souza
  et~al\mbox{.}}{2016}]%
        {DBLP:journals/corr/SouzaCK16}
\bibfield{author}{\bibinfo{person}{Higor~Amario de Souza},
  \bibinfo{person}{Marcos~Lordello Chaim}, {and} \bibinfo{person}{Fabio Kon}.}
  \bibinfo{year}{2016}\natexlab{}.
\newblock \showarticletitle{Spectrum-based Software Fault Localization: {A}
  Survey of Techniques, Advances, and Challenges}.
\newblock \bibinfo{journal}{\emph{CoRR}} (\bibinfo{year}{2016}).
\newblock


\bibitem[\protect\citeauthoryear{Debroy and Wong}{Debroy and Wong}{2011}]%
        {DBLP:conf/sac/DebroyW11}
\bibfield{author}{\bibinfo{person}{Vidroha Debroy} {and}
  \bibinfo{person}{W.~Eric Wong}.} \bibinfo{year}{2011}\natexlab{}.
\newblock \showarticletitle{On the equivalence of certain fault localization
  techniques}. In \bibinfo{booktitle}{\emph{Proceedings of the 2011 {ACM}
  Symposium on Applied Computing (SAC)}}. \bibinfo{pages}{1457--1463}.
\newblock
\urldef\tempurl%
\url{https://doi.org/10.1145/1982185.1982498}
\showDOI{\tempurl}


\bibitem[\protect\citeauthoryear{DiGiuseppe and Jones}{DiGiuseppe and
  Jones}{2011}]%
        {DiGiuseppe:2011:IMF:2001420.2001446}
\bibfield{author}{\bibinfo{person}{Nicholas DiGiuseppe} {and}
  \bibinfo{person}{James~A. Jones}.} \bibinfo{year}{2011}\natexlab{}.
\newblock \showarticletitle{On the Influence of Multiple Faults on
  Coverage-based Fault Localization}. In \bibinfo{booktitle}{\emph{ISSTA}}.
  \bibinfo{publisher}{ACM}, \bibinfo{pages}{210--220}.
\newblock
\showISBNx{978-1-4503-0562-4}


\bibitem[\protect\citeauthoryear{Eric~Wong, Debroy, and Choi}{Eric~Wong
  et~al\mbox{.}}{2010}]%
        {EricWong:2010:FCC:1672348.1672568}
\bibfield{author}{\bibinfo{person}{W. Eric~Wong}, \bibinfo{person}{Vidroha
  Debroy}, {and} \bibinfo{person}{Byoungju Choi}.}
  \bibinfo{year}{2010}\natexlab{}.
\newblock \showarticletitle{A Family of Code Coverage-based Heuristics for
  Effective Fault Localization}.
\newblock \bibinfo{journal}{\emph{JSS}} \bibinfo{volume}{83},
  \bibinfo{number}{2} (\bibinfo{year}{2010}), \bibinfo{pages}{188--208}.
\newblock


\bibitem[\protect\citeauthoryear{Goues, Nguyen, Forrest, and Weimer}{Goues
  et~al\mbox{.}}{2012}]%
        {journals/tse/GouesNFW12}
\bibfield{author}{\bibinfo{person}{Claire~Le Goues}, \bibinfo{person}{ThanhVu
  Nguyen}, \bibinfo{person}{Stephanie Forrest}, {and} \bibinfo{person}{Westley
  Weimer}.} \bibinfo{year}{2012}\natexlab{}.
\newblock \showarticletitle{GenProg: A Generic Method for Automatic Software
  Repair.}
\newblock \bibinfo{journal}{\emph{IEEE Trans. Software Eng.}}
  \bibinfo{volume}{38}, \bibinfo{number}{1} (\bibinfo{year}{2012}),
  \bibinfo{pages}{54--72}.
\newblock


\bibitem[\protect\citeauthoryear{Grbac and Huljeni{\'c}}{Grbac and
  Huljeni{\'c}}{2015}]%
        {grbac2015probability}
\bibfield{author}{\bibinfo{person}{Tihana~Galinac Grbac} {and}
  \bibinfo{person}{Darko Huljeni{\'c}}.} \bibinfo{year}{2015}\natexlab{}.
\newblock \showarticletitle{On the probability distribution of faults in
  complex software systems}.
\newblock \bibinfo{journal}{\emph{Information and Software Technology}}
  \bibinfo{volume}{58} (\bibinfo{year}{2015}), \bibinfo{pages}{250--258}.
\newblock


\bibitem[\protect\citeauthoryear{Groce}{Groce}{2004}]%
        {Groce04errorexplanation}
\bibfield{author}{\bibinfo{person}{Alex Groce}.}
  \bibinfo{year}{2004}\natexlab{}.
\newblock \showarticletitle{Error Explanation with Distance Metrics}. In
  \bibinfo{booktitle}{\emph{TACAS}} \emph{(\bibinfo{series}{LNCS})},
  Vol.~\bibinfo{volume}{2988}. \bibinfo{publisher}{Springer},
  \bibinfo{pages}{108--122}.
\newblock


\bibitem[\protect\citeauthoryear{Henk}{Henk}{2004}]%
        {henk2004understanding}
\bibfield{author}{\bibinfo{person}{Tijms Henk}.}
  \bibinfo{year}{2004}\natexlab{}.
\newblock \bibinfo{title}{Understanding Probability}.
\newblock
\newblock


\bibitem[\protect\citeauthoryear{Janssen, Abreu, and van Gemund}{Janssen
  et~al\mbox{.}}{2009}]%
        {1596502}
\bibfield{author}{\bibinfo{person}{Tom Janssen}, \bibinfo{person}{Rui Abreu},
  {and} \bibinfo{person}{Arjan J.~C. van Gemund}.}
  \bibinfo{year}{2009}\natexlab{}.
\newblock \showarticletitle{Zoltar: a spectrum-based fault localization tool}.
  In \bibinfo{booktitle}{\emph{SINTER}}. \bibinfo{publisher}{ACM},
  \bibinfo{pages}{23--30}.
\newblock


\bibitem[\protect\citeauthoryear{Jaynes}{Jaynes}{2003}]%
        {jaynes03}
\bibfield{author}{\bibinfo{person}{E.~T. Jaynes}.}
  \bibinfo{year}{2003}\natexlab{}.
\newblock \bibinfo{booktitle}{\emph{Probability theory: The logic of science}}.
\newblock \bibinfo{publisher}{Cambridge University Press},
  \bibinfo{address}{Cambridge}.
\newblock


\bibitem[\protect\citeauthoryear{Jha and Seshia}{Jha and Seshia}{2014}]%
        {jhasynt14}
\bibfield{author}{\bibinfo{person}{Susmit Jha} {and} \bibinfo{person}{Sanjit~A.
  Seshia}.} \bibinfo{year}{2014}\natexlab{}.
\newblock \showarticletitle{Are There Good Mistakes? {A} Theoretical Analysis
  of {CEGIS}}. In \bibinfo{booktitle}{\emph{3rd Workshop on Synthesis (SYNT)}}.
  \bibinfo{pages}{84--99}.
\newblock


\bibitem[\protect\citeauthoryear{Jones, Harrold, and Stasko}{Jones
  et~al\mbox{.}}{2002}]%
        {Jones:2002:VTI:581339.581397}
\bibfield{author}{\bibinfo{person}{James~A. Jones}, \bibinfo{person}{Mary~Jean
  Harrold}, {and} \bibinfo{person}{John Stasko}.}
  \bibinfo{year}{2002}\natexlab{}.
\newblock \showarticletitle{Visualization of Test Information to Assist Fault
  Localization}. In \bibinfo{booktitle}{\emph{Proceedings of the 24th
  International Conference on Software Engineering}}
  \emph{(\bibinfo{series}{ICSE '02})}. \bibinfo{publisher}{ACM},
  \bibinfo{pages}{467--477}.
\newblock
\showISBNx{1-58113-472-X}
\urldef\tempurl%
\url{https://doi.org/10.1145/581339.581397}
\showDOI{\tempurl}


\bibitem[\protect\citeauthoryear{Ju}{Ju}{[n. d.]}]%
        {Ju2014}
\bibfield{author}{\bibinfo{person}{Xiaolin et~al. Ju}.} \bibinfo{year}{[n.
  d.]}\natexlab{}.
\newblock  (\bibinfo{year}{[n. d.]}).
\newblock
\urldef\tempurl%
\url{https://doi.org/10.1016/j.jss.2013.11.1109}
\showDOI{\tempurl}


\bibitem[\protect\citeauthoryear{Just, Jalali, and Ernst}{Just
  et~al\mbox{.}}{2014}]%
        {d4j}
\bibfield{author}{\bibinfo{person}{Ren{\'e} Just}, \bibinfo{person}{Darioush
  Jalali}, {and} \bibinfo{person}{Michael~D. Ernst}.}
  \bibinfo{year}{2014}\natexlab{}.
\newblock \showarticletitle{Defects4J: A Database of Existing Faults to Enable
  Controlled Testing Studies for Java Programs} \emph{(\bibinfo{series}{ISSTA
  2014})}. \bibinfo{pages}{437--440}.
\newblock


\bibitem[\protect\citeauthoryear{Kim, Park, and Lee}{Kim et~al\mbox{.}}{2015}]%
        {Kim:2015:NHA:2701126.2701207}
\bibfield{author}{\bibinfo{person}{Jeongho Kim}, \bibinfo{person}{Jonghee
  Park}, {and} \bibinfo{person}{Eunseok Lee}.} \bibinfo{year}{2015}\natexlab{}.
\newblock \showarticletitle{A New Hybrid Algorithm for Software Fault
  Localization}. In \bibinfo{booktitle}{\emph{IMCOM}}.
  \bibinfo{publisher}{ACM}, \bibinfo{pages}{50:1--50:8}.
\newblock
\showISBNx{978-1-4503-3377-1}


\bibitem[\protect\citeauthoryear{Kolmogorov}{Kolmogorov}{1960}]%
        {kolmogorov1960foundations}
\bibfield{author}{\bibinfo{person}{Andrey~N. Kolmogorov}.}
  \bibinfo{year}{1960}\natexlab{}.
\newblock \bibinfo{booktitle}{\emph{Foundations of the Theory of Probability}
  (\bibinfo{edition}{2} ed.)}.
\newblock \bibinfo{publisher}{Chelsea Pub Co}.
\newblock
\urldef\tempurl%
\url{http://www.clrc.rhul.ac.uk/resources/fop/Theory%20of%20Probability%20(small).pdf}
\showURL{%
\tempurl}


\bibitem[\protect\citeauthoryear{Landsberg}{Landsberg}{2016}]%
        {landsberg2016methods}
\bibfield{author}{\bibinfo{person}{David Landsberg}.}
  \bibinfo{year}{2016}\natexlab{}.
\newblock \emph{\bibinfo{title}{Methods and Measures for Statistical Fault
  Localisation (doctoral thesis)}}.
\newblock \bibinfo{thesistype}{Ph.D. Dissertation}. \bibinfo{school}{University
  of Oxford}.
\newblock


\bibitem[\protect\citeauthoryear{Landsberg, Chockler, and Kroening}{Landsberg
  et~al\mbox{.}}{2016}]%
        {Landsbergpfl}
\bibfield{author}{\bibinfo{person}{David Landsberg}, \bibinfo{person}{Hana
  Chockler}, {and} \bibinfo{person}{Daniel Kroening}.}
  \bibinfo{year}{2016}\natexlab{}.
\newblock \showarticletitle{Probabilistic Fault Localisation}. In
  \bibinfo{booktitle}{\emph{Haifa Verification Conference}}.
  \bibinfo{pages}{65--81}.
\newblock


\bibitem[\protect\citeauthoryear{Landsberg, Chockler, Kroening, and
  Lewis}{Landsberg et~al\mbox{.}}{2015}]%
        {Landsberg}
\bibfield{author}{\bibinfo{person}{David Landsberg}, \bibinfo{person}{Hana
  Chockler}, \bibinfo{person}{Daniel Kroening}, {and} \bibinfo{person}{Matt
  Lewis}.} \bibinfo{year}{2015}\natexlab{}.
\newblock \showarticletitle{Evaluation of Measures for Statistical Fault
  Localisation and an Optimising Scheme}.
\newblock In \bibinfo{booktitle}{\emph{FASE}}. \bibinfo{series}{LNCS},
  Vol.~\bibinfo{volume}{9033}. \bibinfo{publisher}{Springer},
  \bibinfo{pages}{115--129}.
\newblock
\showISBNx{978-3-662-46674-2}


\bibitem[\protect\citeauthoryear{Lei, Mao, Dai, and Wang}{Lei
  et~al\mbox{.}}{2012}]%
        {conf/compsac/LeiMDW12}
\bibfield{author}{\bibinfo{person}{Yan Lei}, \bibinfo{person}{Xiaoguang Mao},
  \bibinfo{person}{Ziying Dai}, {and} \bibinfo{person}{Chengsong Wang}.}
  \bibinfo{year}{2012}\natexlab{}.
\newblock \showarticletitle{Effective Statistical Fault Localization Using
  Program Slices.}. In \bibinfo{booktitle}{\emph{COMPSAC}}.
  \bibinfo{publisher}{IEEE Computer Society}, \bibinfo{pages}{1--10}.
\newblock


\bibitem[\protect\citeauthoryear{Liblit, Naik, Zheng, Aiken, and Jordan}{Liblit
  et~al\mbox{.}}{2005}]%
        {Liblit:2005:SSB:1064978.1065014}
\bibfield{author}{\bibinfo{person}{Ben Liblit}, \bibinfo{person}{Mayur Naik},
  \bibinfo{person}{Alice~X. Zheng}, \bibinfo{person}{Alex Aiken}, {and}
  \bibinfo{person}{Michael~I. Jordan}.} \bibinfo{year}{2005}\natexlab{}.
\newblock \showarticletitle{Scalable Statistical Bug Isolation}.
\newblock \bibinfo{journal}{\emph{SIGPLAN Not.}} (\bibinfo{year}{2005}),
  \bibinfo{pages}{15--26}.
\newblock


\bibitem[\protect\citeauthoryear{Liu, Fei, Yan, Han, and Midkiff}{Liu
  et~al\mbox{.}}{2006}]%
        {Liu:2006:SDH:1248726.1248773}
\bibfield{author}{\bibinfo{person}{Chao Liu}, \bibinfo{person}{Long Fei},
  \bibinfo{person}{Xifeng Yan}, \bibinfo{person}{Jiawei Han}, {and}
  \bibinfo{person}{Samuel~P. Midkiff}.} \bibinfo{year}{2006}\natexlab{}.
\newblock \showarticletitle{Statistical Debugging: A Hypothesis Testing-Based
  Approach}.
\newblock \bibinfo{journal}{\emph{IEEE Trans. Softw. Eng.}}
  \bibinfo{volume}{32}, \bibinfo{number}{10} (\bibinfo{year}{2006}),
  \bibinfo{pages}{831--848}.
\newblock


\bibitem[\protect\citeauthoryear{Lucia, Lo, Jiang, Thung, and Budi}{Lucia
  et~al\mbox{.}}{2014}]%
        {LLJTB14}
\bibfield{author}{\bibinfo{person}{Lucia}, \bibinfo{person}{David Lo},
  \bibinfo{person}{Lingxiao Jiang}, \bibinfo{person}{Ferdian Thung}, {and}
  \bibinfo{person}{Aditya Budi}.} \bibinfo{year}{2014}\natexlab{}.
\newblock \showarticletitle{Extended comprehensive study of association
  measures for fault localization}.
\newblock \bibinfo{journal}{\emph{Journal of Software: Evolution and Process}}
  \bibinfo{volume}{26}, \bibinfo{number}{2} (\bibinfo{year}{2014}),
  \bibinfo{pages}{172--219}.
\newblock


\bibitem[\protect\citeauthoryear{Mayer and Stumptner}{Mayer and
  Stumptner}{2008}]%
        {Mayer:2008:EMM:1642931.1642950}
\bibfield{author}{\bibinfo{person}{W. Mayer} {and} \bibinfo{person}{M.
  Stumptner}.} \bibinfo{year}{2008}\natexlab{}.
\newblock \showarticletitle{Evaluating Models for Model-Based Debugging}. In
  \bibinfo{booktitle}{\emph{ASE}}. \bibinfo{pages}{128--137}.
\newblock
\showISBNx{978-1-4244-2187-9}


\bibitem[\protect\citeauthoryear{Moon, Kim, Kim, and Yoo}{Moon
  et~al\mbox{.}}{2014}]%
        {Moon:2014:AMM:2624302.2624536}
\bibfield{author}{\bibinfo{person}{Seokhyeon Moon}, \bibinfo{person}{Yunho
  Kim}, \bibinfo{person}{Moonzoo Kim}, {and} \bibinfo{person}{Shin Yoo}.}
  \bibinfo{year}{2014}\natexlab{}.
\newblock \showarticletitle{Ask the Mutants: Mutating Faulty Programs for Fault
  Localization} \emph{(\bibinfo{series}{ICST '14})}. \bibinfo{pages}{153--162}.
\newblock


\bibitem[\protect\citeauthoryear{Naish and Lee}{Naish and Lee}{2013}]%
        {Naish:Duals}
\bibfield{author}{\bibinfo{person}{Lee Naish} {and} \bibinfo{person}{Hua~Jie
  Lee}.} \bibinfo{year}{2013}\natexlab{}.
\newblock \showarticletitle{Duals in Spectral Fault Localization}. In
  \bibinfo{booktitle}{\emph{Australian Conference on Software Engineering
  (ASWEC)}}. \bibinfo{publisher}{IEEE}, \bibinfo{pages}{51--59}.
\newblock


\bibitem[\protect\citeauthoryear{Naish, Lee, and Ramamohanarao}{Naish
  et~al\mbox{.}}{2011}]%
        {Naish:2011:MSS:2000791.2000795}
\bibfield{author}{\bibinfo{person}{Lee Naish}, \bibinfo{person}{Hua~Jie Lee},
  {and} \bibinfo{person}{Kotagiri Ramamohanarao}.}
  \bibinfo{year}{2011}\natexlab{}.
\newblock \showarticletitle{A Model for Spectra-based Software Diagnosis}.
\newblock \bibinfo{journal}{\emph{ACM Trans. Softw. Eng. Methodol.}}
  (\bibinfo{year}{2011}), \bibinfo{pages}{1--11}.
\newblock


\bibitem[\protect\citeauthoryear{Ochiai}{Ochiai}{1957}]%
        {Ochiai}
\bibfield{author}{\bibinfo{person}{A. Ochiai}.}
  \bibinfo{year}{1957}\natexlab{}.
\newblock \showarticletitle{Zoogeographical Studies on the Soleoid Fishes Found
  in {Japan} and its Neighboring Regions}.
\newblock \bibinfo{journal}{\emph{Bull. Jap. Soc. sci. Fish.}}
  (\bibinfo{year}{1957}), \bibinfo{pages}{526--530}.
\newblock


\bibitem[\protect\citeauthoryear{Papadakis and Le~Traon}{Papadakis and
  Le~Traon}{2015}]%
        {Papadakis:2015:MMF:2858638.2858646}
\bibfield{author}{\bibinfo{person}{Mike Papadakis} {and} \bibinfo{person}{Yves
  Le~Traon}.} \bibinfo{year}{2015}\natexlab{}.
\newblock \showarticletitle{Metallaxis-FL: Mutation-based Fault Localization}.
\newblock \bibinfo{journal}{\emph{Softw. Test. Verif. Reliab.}}
  \bibinfo{volume}{25}, \bibinfo{number}{5-7} (\bibinfo{date}{Aug.}
  \bibinfo{year}{2015}), 24.
\newblock


\bibitem[\protect\citeauthoryear{Parnin and Orso}{Parnin and Orso}{2011}]%
        {PO11}
\bibfield{author}{\bibinfo{person}{Chris Parnin} {and}
  \bibinfo{person}{Alessandro Orso}.} \bibinfo{year}{2011}\natexlab{}.
\newblock \showarticletitle{{Are Automated Debugging Techniques Actually
  Helping Programmers?}}. In \bibinfo{booktitle}{\emph{International Symposium
  on Software Testing and Analysis (ISSTA)}}. \bibinfo{pages}{199--209}.
\newblock


\bibitem[\protect\citeauthoryear{Popper}{Popper}{2005}]%
        {popper2005logic}
\bibfield{author}{\bibinfo{person}{Karl Popper}.}
  \bibinfo{year}{2005}\natexlab{}.
\newblock \bibinfo{booktitle}{\emph{The logic of scientific discovery}}.
\newblock \bibinfo{publisher}{Routledge}.
\newblock


\bibitem[\protect\citeauthoryear{Santelices, Jones, Yu, and Harrold}{Santelices
  et~al\mbox{.}}{2009a}]%
        {5070508}
\bibfield{author}{\bibinfo{person}{R. Santelices}, \bibinfo{person}{J.~A.
  Jones}, \bibinfo{person}{Yanbing Yu}, {and} \bibinfo{person}{M.~J. Harrold}.}
  \bibinfo{year}{2009}\natexlab{a}.
\newblock \showarticletitle{Lightweight fault-localization using multiple
  coverage types}. In \bibinfo{booktitle}{\emph{ICSE}}.
  \bibinfo{pages}{56--66}.
\newblock


\bibitem[\protect\citeauthoryear{Santelices, Jones, Yu, and Harrold}{Santelices
  et~al\mbox{.}}{2009b}]%
        {Santelices:2009:LFU:1555001.1555021}
\bibfield{author}{\bibinfo{person}{Raul Santelices}, \bibinfo{person}{James~A.
  Jones}, \bibinfo{person}{Yanbing Yu}, {and} \bibinfo{person}{Mary~Jean
  Harrold}.} \bibinfo{year}{2009}\natexlab{b}.
\newblock \showarticletitle{Lightweight Fault-localization Using Multiple
  Coverage Types} \emph{(\bibinfo{series}{ICSE '09})}. 11.
\newblock


\bibitem[\protect\citeauthoryear{Steimann and Bertschler}{Steimann and
  Bertschler}{2009}]%
        {DBLP:conf/icst/2009}
\bibfield{author}{\bibinfo{person}{Friedrich Steimann} {and}
  \bibinfo{person}{Mario Bertschler}.} \bibinfo{year}{2009}\natexlab{}.
\newblock \showarticletitle{A Simple Coverage-Based Locator for Multiple
  Faults.}. In \bibinfo{booktitle}{\emph{ICST}} (2009-12-23).
  \bibinfo{publisher}{IEEE Computer Society}, \bibinfo{pages}{366--375}.
\newblock
\showISBNx{978-0-7695-3601-9}


\bibitem[\protect\citeauthoryear{Steimann and Frenkel}{Steimann and
  Frenkel}{2012}]%
        {DBLP:conf/issre/SteimannF12}
\bibfield{author}{\bibinfo{person}{Friedrich Steimann} {and}
  \bibinfo{person}{Marcus Frenkel}.} \bibinfo{year}{2012}\natexlab{}.
\newblock \showarticletitle{Improving Coverage-Based Localization of Multiple
  Faults Using Algorithms from Integer Linear Programming}. In
  \bibinfo{booktitle}{\emph{{ISSRE} November 27-30}}.
  \bibinfo{pages}{121--130}.
\newblock


\bibitem[\protect\citeauthoryear{Steimann, Frenkel, and Abreu}{Steimann
  et~al\mbox{.}}{2013}]%
        {Steimann:2013:TVV:2483760.2483767}
\bibfield{author}{\bibinfo{person}{Friedrich Steimann}, \bibinfo{person}{Marcus
  Frenkel}, {and} \bibinfo{person}{Rui Abreu}.}
  \bibinfo{year}{2013}\natexlab{}.
\newblock \showarticletitle{Threats to the Validity and Value of Empirical
  Assessments of the Accuracy of Coverage-based Fault Locators}. In
  \bibinfo{booktitle}{\emph{ISSTA}}. \bibinfo{publisher}{ACM},
  \bibinfo{pages}{314--324}.
\newblock


\bibitem[\protect\citeauthoryear{Weiser}{Weiser}{1981}]%
        {Weiser:1981:PS:800078.802557}
\bibfield{author}{\bibinfo{person}{Mark Weiser}.}
  \bibinfo{year}{1981}\natexlab{}.
\newblock \showarticletitle{Program Slicing}. In
  \bibinfo{booktitle}{\emph{ICSE}}. \bibinfo{publisher}{IEEE Press},
  \bibinfo{pages}{439--449}.
\newblock


\bibitem[\protect\citeauthoryear{Wong, Debroy, Gao, and Li}{Wong
  et~al\mbox{.}}{2014}]%
        {Dstar}
\bibfield{author}{\bibinfo{person}{W.E. Wong}, \bibinfo{person}{V. Debroy},
  \bibinfo{person}{Ruizhi Gao}, {and} \bibinfo{person}{Yihao Li}.}
  \bibinfo{year}{2014}\natexlab{}.
\newblock \showarticletitle{The {DStar} Method for Effective Software Fault
  Localization}.
\newblock \bibinfo{journal}{\emph{Reliability, IEEE Transactions on}}
  \bibinfo{volume}{63}, \bibinfo{number}{1} (\bibinfo{year}{2014}),
  \bibinfo{pages}{290--308}.
\newblock


\bibitem[\protect\citeauthoryear{Wong, Debroy, Golden, Xu, and
  Thuraisingham}{Wong et~al\mbox{.}}{2012b}]%
        {6058639}
\bibfield{author}{\bibinfo{person}{W.~E. Wong}, \bibinfo{person}{V. Debroy},
  \bibinfo{person}{R. Golden}, \bibinfo{person}{X. Xu}, {and}
  \bibinfo{person}{B. Thuraisingham}.} \bibinfo{year}{2012}\natexlab{b}.
\newblock \showarticletitle{Effective Software Fault Localization Using an RBF
  Neural Network}.
\newblock \bibinfo{journal}{\emph{IEEE Transactions on Reliability}}
  \bibinfo{volume}{61}, \bibinfo{number}{1} (\bibinfo{date}{March}
  \bibinfo{year}{2012}), \bibinfo{pages}{149--169}.
\newblock


\bibitem[\protect\citeauthoryear{Wong, Debroy, Li, and Gao}{Wong
  et~al\mbox{.}}{2012c}]%
        {6258291}
\bibfield{author}{\bibinfo{person}{W.~E. Wong}, \bibinfo{person}{V. Debroy},
  \bibinfo{person}{Y. Li}, {and} \bibinfo{person}{R. Gao}.}
  \bibinfo{year}{2012}\natexlab{c}.
\newblock \showarticletitle{Software Fault Localization Using DStar (D*)}. In
  \bibinfo{booktitle}{\emph{2012 IEEE Sixth International Conference on
  Software Security and Reliability}}. \bibinfo{pages}{21--30}.
\newblock


\bibitem[\protect\citeauthoryear{Wong, Debroy, and Xu}{Wong
  et~al\mbox{.}}{2012a}]%
        {DBLP:journals/tsmc/WongDX12}
\bibfield{author}{\bibinfo{person}{W.~Eric Wong}, \bibinfo{person}{Vidroha
  Debroy}, {and} \bibinfo{person}{Dianxiang Xu}.}
  \bibinfo{year}{2012}\natexlab{a}.
\newblock \showarticletitle{Towards Better Fault Localization: {A}
  Crosstab-Based Statistical Approach}.
\newblock \bibinfo{journal}{\emph{{IEEE} Trans. Systems, Man, and Cybernetics,
  Part {C}}} \bibinfo{volume}{42}, \bibinfo{number}{3} (\bibinfo{year}{2012}),
  \bibinfo{pages}{378--396}.
\newblock
\urldef\tempurl%
\url{https://doi.org/10.1109/TSMCC.2011.2118751}
\showDOI{\tempurl}


\bibitem[\protect\citeauthoryear{Wong, Gao, Li, Abreu, and Wotawa}{Wong
  et~al\mbox{.}}{2016a}]%
        {7390282}
\bibfield{author}{\bibinfo{person}{W.~E. Wong}, \bibinfo{person}{R. Gao},
  \bibinfo{person}{Y. Li}, \bibinfo{person}{R. Abreu}, {and}
  \bibinfo{person}{F. Wotawa}.} \bibinfo{year}{2016}\natexlab{a}.
\newblock \showarticletitle{A Survey on Software Fault Localization}.
\newblock \bibinfo{journal}{\emph{IEEE Transactions on Software Engineering}}
  \bibinfo{number}{99} (\bibinfo{year}{2016}).
\newblock


\bibitem[\protect\citeauthoryear{Wong, Gao, Li, Abreu, and Wotawa}{Wong
  et~al\mbox{.}}{2016b}]%
        {Wong:2016:SSF:3012168.3012182}
\bibfield{author}{\bibinfo{person}{W.~Eric Wong}, \bibinfo{person}{Ruizhi Gao},
  \bibinfo{person}{Yihao Li}, \bibinfo{person}{Rui Abreu}, {and}
  \bibinfo{person}{Franz Wotawa}.} \bibinfo{year}{2016}\natexlab{b}.
\newblock \showarticletitle{A Survey on Software Fault Localization}.
\newblock \bibinfo{journal}{\emph{IEEE Trans. Softw. Eng.}}
  \bibinfo{volume}{42}, \bibinfo{number}{8} (\bibinfo{date}{Aug.}
  \bibinfo{year}{2016}), 34.
\newblock


\bibitem[\protect\citeauthoryear{Wong and Qi}{Wong and Qi}{2006}]%
        {journals/jss/WongQ06}
\bibfield{author}{\bibinfo{person}{W.~Eric Wong} {and} \bibinfo{person}{Yu
  Qi}.} \bibinfo{year}{2006}\natexlab{}.
\newblock \showarticletitle{Effective program debugging based on execution
  slices and inter-block data dependency.}
\newblock \bibinfo{journal}{\emph{JSS}} (\bibinfo{year}{2006}),
  \bibinfo{pages}{891--903}.
\newblock


\bibitem[\protect\citeauthoryear{Wong, Qi, Zhao, and Cai}{Wong
  et~al\mbox{.}}{2007}]%
        {Wong:2007:EFL:1299135.1299726}
\bibfield{author}{\bibinfo{person}{W.~Eric Wong}, \bibinfo{person}{Yu Qi},
  \bibinfo{person}{Lei Zhao}, {and} \bibinfo{person}{Kai-Yuan Cai}.}
  \bibinfo{year}{2007}\natexlab{}.
\newblock \showarticletitle{Effective Fault Localization Using Code Coverage}.
  In \bibinfo{booktitle}{\emph{COMPSAC}}. \bibinfo{pages}{449--456}.
\newblock


\bibitem[\protect\citeauthoryear{Wotawa, Stumptner, and Mayer}{Wotawa
  et~al\mbox{.}}{2002}]%
        {Mayerold}
\bibfield{author}{\bibinfo{person}{Franz Wotawa}, \bibinfo{person}{Markus
  Stumptner}, {and} \bibinfo{person}{Wolfgang Mayer}.}
  \bibinfo{year}{2002}\natexlab{}.
\newblock \showarticletitle{Model-Based Debugging or How to Diagnose Programs
  Automatically}.
\newblock In \bibinfo{booktitle}{\emph{Developments in Applied Artificial
  Intelligence}}. \bibinfo{series}{LNCS}, Vol.~\bibinfo{volume}{2358}.
  \bibinfo{pages}{746--757}.
\newblock
\showISBNx{978-3-540-43781-9}


\bibitem[\protect\citeauthoryear{Xie, Chen, Kuo, and Xu}{Xie
  et~al\mbox{.}}{2013}]%
        {Xie}
\bibfield{author}{\bibinfo{person}{Xiaoyuan Xie}, \bibinfo{person}{Tsong~Yueh
  Chen}, \bibinfo{person}{Fei-Ching Kuo}, {and} \bibinfo{person}{Baowen Xu}.}
  \bibinfo{year}{2013}\natexlab{}.
\newblock \showarticletitle{A Theoretical Analysis of the Risk Evaluation
  Formulas for Spectrum-based Fault Localization}.
\newblock \bibinfo{journal}{\emph{ACM Trans. Softw. Eng. Methodol.}}
  (\bibinfo{year}{2013}), \bibinfo{pages}{31:1--31:40}.
\newblock


\bibitem[\protect\citeauthoryear{Xuan and Monperrus}{Xuan and
  Monperrus}{2014}]%
        {xuan:hal-01018935}
\bibfield{author}{\bibinfo{person}{Jifeng Xuan} {and} \bibinfo{person}{Martin
  Monperrus}.} \bibinfo{year}{2014}\natexlab{}.
\newblock \showarticletitle{{Learning to Combine Multiple Ranking Metrics for
  Fault Localization}}. In \bibinfo{booktitle}{\emph{ICSME}}.
\newblock
\urldef\tempurl%
\url{https://doi.org/10.1109/ICSME.2014.41}
\showDOI{\tempurl}


\bibitem[\protect\citeauthoryear{Yilmaz and Williams}{Yilmaz and
  Williams}{2007}]%
        {Yilmaz:2007:AMD:1321631.1321659}
\bibfield{author}{\bibinfo{person}{Cemal Yilmaz} {and} \bibinfo{person}{Clay
  Williams}.} \bibinfo{year}{2007}\natexlab{}.
\newblock \showarticletitle{An Automated Model-based Debugging Approach}. In
  \bibinfo{booktitle}{\emph{ASE}}. \bibinfo{publisher}{ACM},
  \bibinfo{pages}{174--183}.
\newblock


\bibitem[\protect\citeauthoryear{Yoo}{Yoo}{2012}]%
        {Yoo12}
\bibfield{author}{\bibinfo{person}{Shin Yoo}.} \bibinfo{year}{2012}\natexlab{}.
\newblock \showarticletitle{Evolving Human Competitive Spectra-Based Fault
  Localisation Techniques}. In \bibinfo{booktitle}{\emph{SSBSE}}
  \emph{(\bibinfo{series}{LNCS})}, Vol.~\bibinfo{volume}{7515}.
  \bibinfo{pages}{244--258}.
\newblock


\bibitem[\protect\citeauthoryear{Yoo, Xiaoyuan, Kuo, Chen, Yueh, and
  Harman}{Yoo et~al\mbox{.}}{2014}]%
        {Nopotofgold}
\bibfield{author}{\bibinfo{person}{S. Yoo}, \bibinfo{person}{X. Xiaoyuan},
  \bibinfo{person}{F. Kuo}, \bibinfo{person}{T. Chen},
  \bibinfo{person}{Y.~Tsong Yueh}, {and} \bibinfo{person}{M. Harman}.}
  \bibinfo{year}{2014}\natexlab{}.
\newblock \showarticletitle{No pot of gold at the end of program spectrum
  rainbow: Greatest risk evaluation formula does not exist.}
\newblock \bibinfo{journal}{\emph{Department of Computer Science, University
  College London}} (\bibinfo{year}{2014}).
\newblock


\bibitem[\protect\citeauthoryear{Zhang, He, Gupta, and Gupta}{Zhang
  et~al\mbox{.}}{2005}]%
        {Zhang:2005:EEU:1085130.1085135}
\bibfield{author}{\bibinfo{person}{Xiangyu Zhang}, \bibinfo{person}{Haifeng
  He}, \bibinfo{person}{Neelam Gupta}, {and} \bibinfo{person}{Rajiv Gupta}.}
  \bibinfo{year}{2005}\natexlab{}.
\newblock \showarticletitle{Experimental Evaluation of Using Dynamic Slices for
  Fault Location} \emph{(\bibinfo{series}{AADEBUG})}. \bibinfo{publisher}{ACM},
  \bibinfo{pages}{33--42}.
\newblock


\end{thebibliography}

\newpage

\appendix

\section{Proofs}

In this appendix, we present the proofs supporting the main text. To simplify, we have put a proof later in our order of presentation if a part of that proof relies on a part of an earlier proof. 

To aid in the proofs we introduce some 
notation. For each $m^j \in M$, $m^j_{i,k}$ is the value at the $i$th column and $k$th row. $m^j_k$ is the matrix consisting of the $k$th row of $m^j$. For each $t_k \in T$ we let $M_k$ = \{$m^j_k| 1 \leq j \leq |M|$\}. Intuitively, this is all different $k$th rows of all models. 
We make use of a concatonation operation $\cdot$ such that $x$ $\cdot$ $y$ is the concatonation of two matrices, and extend the definition such that  $M_i \cdot M_j$ = \{$x \cdot y | x \in M_i  \wedge y \in M_j$\}. 
Accordingly, definition 3 is designed to conform to the following assumption $M = M_1 \cdot M_2 \cdot ... \cdot M_{|T|}$. It is observed that $|M_k|$ = $2^{\rho_k}-1$ if $t_k$ is failing, and 1 otherwise.

\begin{proposition}~\label{a_causal}
Equation~\ref{proof_eq2} follows using the defs. 
\end{proposition}

\begin{proof}
We must show $P(H_i \wedge u_i)  = P(H_i)$. It is sufficient to show $v_k(H_i \wedge u_i)$  = $v_k(H_i)$. Thus, it is sufficient to show $v_k(H_i \wedge u_i)$ = $v_k(H_i) \cap v_k(u_i)$ (by def.~\ref{def_semantics}). Now, $v_k(H_i) \cap v_k(u_i) \subseteq v_k(H_i)$ (as in general $X \cap Y \subseteq X$), thus it remains to show $v_k(H_i) \cap v_k(u_i)$  $\subseteq$ $v_k(H_i)$.  It is sufficient to show  $v_k(u_i)$  $\subseteq$ $v_k(H_i)$ (as in general $X \subseteq X$). Now, 
$v_k(H_i)$ 
= (by the abbreviation for $H_i$)
$v_k(h_1 \wedge \neg h_2 \wedge ... \wedge \neg h_{|U|-1})$ 
= (by def.~\ref{def_semantics})
$v_k(h_1) \cap v_k(\neg h_2) \cap ... \cap \neg v_k(h_{|U|-1})$ 
=
$\{m^j|m^j_{1,k} = 2\} \cap \{m^j|m^j_{1,k} \in \{1,0\}\} \cap \dots  $ (by def. 3 and 2). 
=
$\{m^j|m^j_{1,k} = 2\} \cap \{m^j|m^j_{1,k} \in \{1,0\}\} \cap \dots  $
=
$\{m^j|m^j_{1,k} = 2 \wedge m^j_{1,k} \in \{1,0\} \wedge \dots \} $
=
$\{m^j|m^j_{1,k} \in \{1,2\} \wedge m^j_{1,k} = 2 \wedge m^j_{1,k} \in \{1,0\} \wedge \dots \} $.
A subset of which is
$\{m^j|m^j_{1,k} \in \{1,2\}\}$, which is $v_k(u_i)$ (by def. 3). 

\begin{proposition}~\label{a_causal}
Equation~\ref{proof_eq4} follows using the defs. 
\end{proposition}

We sketch the proof for equation~\ref{proof_eq4}.
$P_t(H_i)$ is equal to $\frac{|v_t(H_i)|}{|M|}$ (by def.~\ref{def_probability}). 
We first do the top condition of equation~\ref{proof_eq4}. Without loss of generality, let $k = 1$.
Assume $c_{i,1} = e_1$ = 1. 
Now, $M = M_1 \cdot M_2 \cdot ... \cdot M_{|M|}$ (by A1). 
 Thus,
$v_1(H_i)$ = $\{m^j_1 |m^j \in v_1(H_i)\} \cdot M_2 \cdot ... \cdot M_n$ (as in general, $v_1(\phi)$ = $\{m^j_k \in M_1|m^j \in v_k(\phi)\} \cdot M_2 \cdot ... \cdot Mn$). So,
$|v_1(H_i)|$ = $|\{m^j_1|m^j \in v_1(H_i)\}| \times |M_2| \times ... \times |M_{|T|}|$. So,
$|v_1(H_i)|$ = $1  \times |M_2| \times ... \times |M_{|T|}|$ (given if $m^j,m^k \in $ $v_1(H_i)$ then $m^j_1 = m^k_1$).
Thus,
$1 \times |M_2| \times ... \times |M_{|T|}|$ / $|M_1| \times |M_2| \times ... \times |M_{|T|}|$.
This is equal to $1/|M_1|$ (by cancellation), which is equal to $1/(2^{\rho_1}-1)$ (given $\rho_1 = \sum_{i < |U|} c_{i,1}$). 
We now do the bottom condition (when either $c_{i,1}$ or $e_1$ is 0). Accordingly, there are no models in $v_1(h_i)$ (by def.~\ref{Causal Models}), thus no models in $v_1(H_i)$ (by def.~\ref{def_semantics} and $H_i$), and so $P_1(H_i) = 0$. %\qed 

\begin{proposition}~\label{a_causal}
Equation~\ref{proof_eq5} follows using the defs. 
\end{proposition}

We must show $P(u_i)$ = $\sum_k c_{i,k}/|T|$. $P(u_i)$ is equal to $f(u_i)/|T|$ (by proposition~\ref{prop_basic_lang}).  The latter is equal to $\sum_k f_k(u_i)/|T|$. It remains to show $f_k(u_i) = c_{i,k}$, for both cases when $c_{i,k}$ is 1 or 0 (given $c$ is a Boolean matrix). Assume $c_{i,k}$ is 1. Then for all $m \in M$, $m_{i,k} \in {1, \bullet}$ (by def.~\ref{def_causal_models}). Thus, $v_k(u_i) = M$ (by def.~\ref{def_semantics}). So, $f_k(u_i)$ = 1. Assume $c_{i,k}$ is 0. Then for all $m \in M$, $m_{i,k} = 0$ (by def.~\ref{def_causal_models}). Thus, $v_k(u_i) = M$ (by def.~\ref{def_semantics}). 

\end{proof}

\begin{proposition}~\label{a_L*}
For all $\phi \in L^*$, $P(\phi)$ = $\frac{f(\phi)}{|T|}$
\end{proposition}

\begin{proof}
We must show for all $\phi \in L^*$, $P(\phi)$ = $\frac{f(\phi)}{|T|}$.
$P(\phi)$ is equal to $\sum_k P_k(\phi)/|T|$ (by def.\ref{def_probability}). The latter is equal to $\sum_k (w(v_k(\phi))$ $/w(M))/|T|$ (by def.\ref{def_probability}).
Now, $\frac{f(\phi)}{|T|}$ = $\sum_k \frac{f_k(\phi)}{|T|}$ (by definition of $f$). Thus it is remains to prove $w(v_k(\phi))/w(M)) = f_k(\phi)$.
We have two cases to consider, when $f_k(\phi) = 1$  and when it is 0. We do the former first. Assume $f_k(\phi) = 1$, then 
$v_k(u_i) = M$ (by definition of $f$). So $w(v_k(\phi))/w(M))$ = $w(M)/w(M)$ = 1. Assume $f_k(\phi) = 0$, then $v_k(\phi) \neq M$. If the latter, then $v_k(\phi) = \emptyset$ (as $\phi \in L^*$). So $w(v_k(\phi))/w(M))$ = $w(\emptyset)/w(M)$ = 0 (by def. of $w$). 
\end{proof}

\begin{proposition}~\label{a_classical}
Equation~\ref{eq_classical_prob} follows given indifference
\end{proposition}

\begin{proof}
Let $v_k(\phi) = \{m^i, \dots,m^j\}$ and $M = \{m^1, \dots ,M^{|M|}\}$. 
Then $P_k(\phi)$ = $\sum_{i=1}^{j} w(\{m^i\})$ / $\sum_{k=1}^{|M|} w(\{m^k\})$ (by Def.~\ref{def_probability}). Now, $w(\{m^k\}) > 0$ for all $m^k \in M$ (by the conditions on $w$ and indifference). Given indifference, let $w(\{m^k\})$ = $x \times |\{m^k\}|$. Then $P_k(\phi)$ = $\sum_{i=1}^{j} x|\{m^i\}|$ / $\sum_{k=1}^{|M|} x|\{m^k\}|$ (by substitution). Thus, $P_k(\phi)$ = $x\sum_{i=1}^{j} |\{m^i\}|$ / $x\sum_{k=1}^{|M|} |\{m^k\}|$ (by distribution). So, $P_k(\phi)$ = $\sum_{i=1}^{j} |\{m^i\}|$ / $\sum_{k=1}^{|M|} |\{m^k\}|$ (by cancellation). Equivalently, $P_k(\phi)$ = $|\{m^i,\dots,m^j\}|$ / $|M|$. So,  $P_k(\phi)$ = $|v_k(\phi)|$ / $|M|$ (by substitution). 
\end{proof}

\vspace{0.5cm}

\begin{proposition}~\label{a_updating}
Let $c$ be a coverage matrix where $c_{i,k} = c_{j,k} = 1$ for some $t_k \in T$. Then $P(H_i|u_i) < P(H_i|u_i \wedge \neg h_j)$.
\end{proposition}

\begin{proof}
We show $P(H_i|u_i) < P(H_i|u_i \wedge \neg h_j)$ given the conditions of the proposition. 
This is equivalent to $P(H_i \wedge u_i)/P(u_i) < P(H_i \wedge u_i \wedge \neg h_j)/$ $P(u_i \wedge \neg h_j)$ (by $|$-def.). 
Equivalently, $P(H_i \wedge u_i)/P(u_i) < P(H_i \wedge u_i)/$ $P(u_i \wedge \neg h_j)$  (given $H_i$ implies $\neg h_j$). 
Accordingly, it is sufficient to show $1/P(u_i) < 1/P(u_i \wedge \neg h_j)$. 
It then suffices to show $P(u_i) > P(u_i \wedge \neg h_j)$. 
Equivalently, $\sum_k P_k(u_i)/|T| > \sum_k P_k(u_i \wedge \neg h_j)/|T|$ (by def.~\ref{def_probability}).
Equivalently, $\sum_k P_k(u_i) > \sum_k P_k(u_i \wedge \neg h_j)$ (by cancellation). 
Equivalently, $\sum_k |v_k(u_i)|/|M| > \sum_k |v_k(u_i \wedge \neg h_j)|/|M|$.
Equivalently, $\sum_k |v_k(u_i)| > \sum_k |v_k(u_i \wedge \neg h_j)|$ (by cancellation). 
Equivalently, $\sum_k |v_k(u_i)| > \sum_k |v_k(u_i) \cap (M - v(h_j)))|$ (by def.~\ref{def_language}). 
It is sufficient to show $M - v(h_j) \subset$ $v_k(u_i)$. $v_k(u_i) = M$ (given $c_{i,k} = 1$). 
Thus it suffices to show 
$M - v(h_j) \subset$ $M$. 
To prove this, it is sufficient to show $v(h_j) \neq \emptyset$. This holds given $c_{j,k} = 1$ (given the conditions of the proposition and def.~\ref{def_causal_models}).
\end{proof}

\begin{proposition}~\label{a_fault_li}
Equations~\ref{f1},~\ref{f2}, \& ~\ref{f3} follow given indifference.
\end{proposition}

\begin{proof}
We sketch the proof here. Equation~\ref{f1} follows given the definition of $f_i$ in Section~\ref{section_classical_foundations}. Equation~\ref{f3} is similar to the proof for Equation~\ref{proof_eq4}. We do the proof for Equation~\ref{f2} here. 
As the general proof for this involves syntactically complex formulae, we sketch the proof for reasons of space, and do the case for when there are two test cases $T = \{t_1, t_2\}$ here. 
$P(\Diamond_1 h_i \vee \Diamond_2 h_i)$ = 
$(1 - P(\Diamond_1 h_i))$ $( P(\Diamond_1 h_i)) + P(\Diamond_2 h_i))$. The latter expression is equal to 
$( P(\Diamond_1 h_i)) + P(\Diamond_2 h_i))$ - $(P(\Diamond_1 h_i) P(\Diamond_1 h_i))$. Thus it is sufficient to prove $P(\Diamond_1 h_i \vee \Diamond_2 h_i)$ is equal to that. Now, in general $P(\phi \vee \psi) = P(\phi) + P(\psi) - P(\phi \wedge \psi)$. Thus it suffices to show $P(\Diamond_1 h_i) P(\Diamond_1 h_i)$ = $P(\Diamond_1 h_i \wedge \Diamond_1 h_i)$. It is sufficient to show $|v_1(h_i) \cap v_2(h_i)|/|M|$ = $|v_1(h_i)|/|M| \times |v_2(h_i)|/|M|$. It is sufficient to show $|v_1(h_i) \cap v_2(h_i)|$ = $|v_1(h_i)||v_2(h_i)|$. Now $M = M_1 \cdot M_2$. Thus, for every $m^j_1 \in M_1$ there is a paired $m^j_2 \in M_2$. Thus $|v_1(h_i) \cap v_2(h_i)|$ = $|v_1(h_i)||v_2(h_i)|$.
\end{proof}

\end{document}